\documentclass[final,journal,twoside]{IEEEtran}

\ifCLASSINFOpdf
  \usepackage[pdftex]{graphicx}
\else
  \usepackage[dvips]{graphicx}
\fi

\usepackage{amsmath,amsthm,amsfonts,amssymb,esint,mathabx,cite}

\theoremstyle{plain} 
\newtheorem{theorem}{Theorem}
\newtheorem{lemma}[theorem]{Lemma}
\theoremstyle{definition}  
\newtheorem{definition}{Definition}
\newtheorem{assumptions}{Assumptions}
\newtheorem{example}{Example}
\theoremstyle{remark}  
\newtheorem{remark}{Remark}
\DeclareMathOperator{\sech}{sech}
\DeclareMathOperator{\diag}{diag}
\DeclareMathOperator{\tr}{tr}
\DeclareMathOperator{\sinc}{sinc}
\DeclareMathOperator{\sgn}{sgn}

\newcommand{\eg}{\emph{e.g.}}
\newcommand{\ie}{\emph{i.e.}}
\newcommand{\const}[1]{{\mathcal{#1}}}
\newcommand{\Reals}{\mathbb{R}}
\newcommand{\Complex}{\mathbb{C}}
\newcommand{\E}{\mathsf{E}}
\newcommand{\inners}[2]{{\langle{#1},{#2}\rangle}_s}
\newcommand{\inner}[2]{{\langle{#1},{#2}\rangle}}
\newcommand{\der}{\mathrm{d}}
\newcommand{\eqdef}{\stackrel{\Delta}{=}}

\begin{document}

\title{Information Transmission using
the Nonlinear Fourier Transform, Part I:\\Mathematical
Tools\thanks{\noindent Submitted for publication on February 16,
2012;  revised February 15, 2013; accepted March 21, 2013.
The material in this paper was presented in part at the 2013 IEEE International Symposium on 
Information Theory. The authors are with the Edward S. Rogers Sr.\ Dept.\ of
Electrical and Computer Engineering, University of Toronto,
Toronto, ON M5S 3G4, Canada. Email:
\texttt{\{mansoor,frank\}@comm.utoronto.ca.}}}

\markboth{IEEE Transactions on Information Theory}{Yousefi and
Kschischang}

\author{Mansoor~I.~Yousefi and Frank~R.~Kschischang,~\IEEEmembership{Fellow,~IEEE}}

\IEEEpubid{0000--0000/00\$00.00~\copyright~2014 IEEE}

\maketitle

\begin{abstract}
\ifCLASSOPTIONonecolumn\relax\else\boldmath\fi
The nonlinear Fourier transform (NFT), a powerful tool in
soliton theory and exactly solvable models, is a method for
solving integrable partial differential equations governing
wave propagation in certain nonlinear media.  The NFT
decorrelates signal degrees-of-freedom in such models, in
much the same way that the Fourier transform does for linear
systems.  In this three-part series of papers, this
observation is exploited for data transmission over
integrable channels such as optical fibers, where pulse
propagation is governed by the nonlinear Schr\"odinger
equation.  In this transmission scheme, which can be viewed
as a nonlinear analogue of orthogonal frequency-division
multiplexing commonly used in linear channels, information
is encoded in the nonlinear frequencies and their spectral
amplitudes.  Unlike most other fiber-optic transmission
schemes, this technique deals with both dispersion and
nonlinearity directly and unconditionally without the need
for dispersion or nonlinearity compensation methods.  This
first paper explains the mathematical tools that underlie
the method.
\end{abstract}


\maketitle

\section{Introduction}
\label{sec:introduction}

\IEEEPARstart{A}{n evolution} equation is a partial differential
equation (PDE) for an unknown function $q(t,z)$ of the form
\begin{equation}
q_z = K(q),
\label{eqn:evolution}
\end{equation}
where $K(q)$ is an expression involving $q$ and its derivatives
with respect to $t$.  In (\ref{eqn:evolution}) and throughout this
paper, subscripts are used to denote partial derivatives with respect
to the corresponding variable.

The evolution equation (\ref{eqn:evolution}) is a PDE in
$1+1$ dimensions, \ie, one variable $t \in \mathbb{R}$ represents a
temporal dimension and one variable $z \geq 0$ represents a spatial
dimension.  In most mathematical literature, the roles of $z$ and $t$
are interchanged, so rather than a spatial evolution as in
(\ref{eqn:evolution}), mathematicians study temporal evolution.

An important example of an evolution equation, and the one that
motivates this work, is the stochastic nonlinear Schr\"odinger (NLS)
equation governing the complex envelope of a narrowband signal in an
optical fiber in the presence of noise, given by
\cite{zakharov1972ett,agrawal2001nfo}
\begin{IEEEeqnarray}{rCl}
\frac{\partial Q(\tau,\ell)}{\partial \ell}
  &=&\frac{j\beta_2}{2}\frac{\partial^2 Q(\tau,\ell)}{\partial \tau^2}
- j\gamma Q(\tau,\ell)|Q(\tau,\ell)|^{2} \nonumber \\
&  & +\: N(\tau,\ell), \quad 0 \leq \ell \leq\const{L}.
\label{eq:st-nls}
\end{IEEEeqnarray}
Here $\ell$ denotes distance (in km) along the fiber; the
transmitter is located at $\ell = 0$, and the receiver is located
at $\ell = \const{L}$.  The symbol $\tau$ represents retarded
time, measured in seconds, \ie, $\tau=t-\beta_1 \ell$ where $t$
is ordinary time and $\beta_1$ is a constant, and $Q(\tau,\ell)$
is the complex envelope of the signal propagating in the fiber.
The coefficient $\beta_2$, measured in ${\rm s^2/km}$, is called
the chromatic dispersion coefficient, while $\gamma$, measured in
${\rm W^{-1}km^{-1}}$, is the nonlinearity parameter.  Finally,
$N(\tau,\ell)$ is bandlimited white Gaussian noise with in-band
spectral density $\sigma^2_0$ (${\rm W/(km}\cdot{\rm Hz)}$) and
autocorrelation
\begin{IEEEeqnarray*}{rCl}
\E \left\{N(\tau,\ell)N^*(\tau^\prime,\ell^\prime)\right\}
=\sigma^2_0\delta_W(\tau-\tau^\prime)\delta(\ell-\ell^\prime),
\end{IEEEeqnarray*}
where $\delta_W(x)= 2W\sinc(2Wx)$.  It is assumed that
$Q(\tau,\ell)$ is bandlimited to $W$ for all $\ell$, $0 \leq \ell
\leq \const{L}$.

The stochastic NLS equation (\ref{eq:st-nls}) models both
chromatic dispersion (captured by the $j\beta_2\partial^2
Q/\partial \tau^2$ term), which is responsible for temporal
broadening, and the Kerr nonlinearity (captured by the $\gamma
|Q|^2 Q$ term), which is responsible for spectral broadening.
Pulse propagation is governed by the tension between these
effects and can be linearly dominated, nonlinearly dominated, or
solitonic (in which case the two effects are balanced).  The NLS
equation defines a nonlinear dispersive waveform channel from
$Q(\tau,0)$ at the transmitter to $Q(\tau,\const{L})$ at the
receiver.

\IEEEpubidadjcol

Such a nonlinear dispersive waveform channel is a major departure
from the classical additive white Gaussian noise and wireless
fading channels in terms of analytical difficulty.  Here the
signal degrees-of-freedom couple together via the nonlinearity
and dispersion in a complicated manner, making it difficult to
establish the channel input-output map, even deterministically.
Most current approaches assume a linearly-dominated regime of
operation, consider the nonlinearity as a small perturbation, or
are geared towards managing and suppressing the (detrimental)
effects of the nonlinear and dispersive terms.  In-line dispersion
management, digital backpropagation, and other forms of
electronic pre-and post-compensation belong to this class of
methods (see \cite{essiambre2010clo,
  kramer2003sec,ip2008cdn,roberts2006epo,beygi2011}
and references therein).

In this paper we adopt a different philosophy.  Rather than
treating nonlinearity and dispersion as nuisances, we seek a
transmission scheme that is fundamentally compatible with these
effects.  We effectively ``diagonalize'' the nonlinear
Schr\"odinger channel with the help of the nonlinear Fourier
transform (NFT), a powerful tool for solving \emph{integrable}
nonlinear dispersive partial differential equations
\cite{ablowitz2006sai,faddeev2007hmi}. The NFT uncovers linear
structure hidden in the one-dimensional cubic nonlinear
Schr\"odinger equation, and can be viewed as a generalization of
the (ordinary) Fourier transform to certain nonlinear systems.

With the help of the nonlinear Fourier transform, we are able to
represent a signal by its discrete and continuous nonlinear
spectra.  While the signal propagates along the fiber based on
the complicated NLS equation, the action of the channel on its
spectral components is given by simple independent
linear equations.  Just as the (ordinary) Fourier transform
converts a linear convolutional channel $y(t)=x(t)\convolution h(t)$ into a
number of parallel scalar channels, the nonlinear Fourier
transform converts a nonlinear dispersive channel described by a
\emph{Lax convolution} (see Sec.~\ref{sec:laxform}) into a number of
parallel scalar channels.  This suggests that information can be
encoded (in analogy with orthogonal frequency-division
multiplexing) in the nonlinear spectra.

The nonlinear Fourier transform is intertwined with the existence
of soliton solutions to the NLS equation.  Solitons are pulses
that retain their shape (or return periodically to their initial
shape) during propagation, and can be viewed as system
eigenfunctions, similar to the complex exponentials $e^{j\omega
t}$, which are eigenfunctions of linear systems.  An arbitrary
waveform can be viewed as a combination of solitons, associated
with the discrete nonlinear spectrum, and a non-solitonic
(radiation) component, associated with the continuous nonlinear
spectrum.

The goal of this first article is to introduce the mathematical
tools that underlie this approach to information transmission.
These tools are sufficiently general to encompass not only the
nonlinear Schr\"odinger equation, but also other completely
integrable nonlinear dispersive PDEs.  Thus, the transmission
scheme described here can also be applied to any channel model in
this general class.  These tools are also described in
mathematics and physics (see, \eg,
\cite{tao2006nfa,ablowitz2006sai,faddeev2007hmi}); here we
attempt to extract those aspects of the theory that are relevant
to the engineering aspects of the information transmission
problem.  In Part~II \cite{yousefi2012nft2} we will provide 
numerical methods for forward NFT at the receiver and in 
Part~III \cite{yousefi2013nft3} describe a nonlinear
frequency-division multiplexing communication method based on NFT, provide 
algorithms to implement inverse NFT at the transmitter and give
examples of achievable spectral efficiencies in actual fiber-optic systems.

After submitting this paper, we
became aware of the related paper
of Hasegawa and Nyu \cite{hasegawa1993ec}
introducing \emph{eigenvalue communication}, in which the authors
propose
to
encode information in quantities conserved under
NLS propagation, and use
the inverse scattering transform to
decode the amplitude of an isolated soliton-like pulse
$A\sech(t)$ in a single-user, point-to-point channel.
Using conserved
quantities can be helpful, as it has the potential to simplify communication
system design.  However, no improvement in capacity is achieved relative to
a scheme which encodes information in an equivalent, but not necessarily
conserved, set of parameters (\eg, modulating amplitude and phase and using
backpropagation for decoding).

Our motivation for introducing transmission schemes based on
the NFT stems---not from the simplicity arising through the
use of conserved quantities (as in \cite{hasegawa1993ec}) but rather---from
our observation that the use of conventional linear
multiplexing in nonlinear
channels inevitably leads to an interference-limited multiuser system \cite{yousefi2013nft3}.
As a result, a major factor contributing to the capacity limitations in optical fiber
networks \cite{essiambre2010clo} is the use of wavelength-division multiplexing.
By introducing nonlinear frequency-division multiplexing (NFDM) based on the
NFT (exploiting the integrability of the NLS equation),
a transmission scheme without deterministic distortions
can result; see [Part~III, particularly Sections VI. F and VI. G] for further
details about NFDM in comparison with \cite{hasegawa1993ec}.
Thus we use the tools of integrable systems for a different purpose
than in \cite{hasegawa1993ec}.

We will find it convenient to work with the
nonlinear Schr\"odinger equation (\ref{eq:st-nls}) in a normalized form.
By changing variables
\begin{IEEEeqnarray*}{rCl}
q=\frac{Q}{\sqrt{P}}, \qquad z=\frac{l}{\const{L}},\qquad t=\frac{\tau}{T_{0}},
\end{IEEEeqnarray*}
with $T_{0}=\sqrt{|\beta_2|\const{L}/2}$ and $P= 2 /(\gamma \const{L})$, we get the normalized NLS equation
\begin{IEEEeqnarray}{rCl}
jq_z(t,z)=q_{tt}+2|q(t,z)|^2q(t,z) + n(t,z),
\label{eq:nnls}
\end{IEEEeqnarray}
where 
\[
\E \left\{n(t,z)n^*(t',z')\right\}
=\frac{\sigma^2_0\const{L}}{PT_0} \delta_{W_n}(t-t')\delta(z-z'),
\]
where $W_n=WT_0$ is the normalized bandwidth. Throughout this paper, \eqref{eq:nnls} will be our primary
and motivating illustrative example.

\section{A Brief History of the Nonlinear Fourier Transform}

The nonlinear Fourier transform (also known as the inverse
scattering transform or IST) was originally developed as a
method for solving certain
nonlinear dispersive partial differential equations.
These are integrable PDEs, \ie, nonlinear differential
equations exhibiting certain hidden linearity. There are several
integrable equations having physical significance, among which is
the NLS equation.  The IST method was a result of extensive
efforts in theoretical physics and applied mathematics in 1960s and later,
closely associated with the notion of solitons in integrable models
\cite{ablowitz2006sai,faddeev2007hmi}. 

In the 1950s, in one of the first dynamical-systems simulations
performed on a computer \cite{fermi1955snp}, Fermi, Pasta and
Ulam performed a numerical experiment to understand why solids
have finite heat conductivity.  They modeled the solid as a
lattice with point masses at the lattice points coupled with
springs each having a quadratic nonlinearity.  To their surprise,
rather than observing an equipartition of energy among all Fourier modes,
energy cycled periodically among a few low-order modes.  Such
behavior implies that the nonlinear oscillator behaves somehow
linearly.

In the 1960s, Zabusky and Kruskal showed that 
equation of motion for the Fermi-Pasta-Ulam lattice in the continuum
limit is a remarkable PDE called
the Korteweg-de~Vries (KdV) equation
\cite{zabusky1965isc}, known in the study
of water waves.
The KdV
equation for the
evolution of a
real-valued pulse $q(t,z)$ as a function of time $t$ and
distance $z$ is
\begin{IEEEeqnarray}{rCl} 
q_{z}=qq_{t}+q_{ttt}.
\label{eq:kdv}
\end{IEEEeqnarray}
Zabusky and Kruskal found that (\ref{eq:kdv}) has
pulse-like (localized) solutions whose shape is preserved (or varies
periodically) during propagation.  Furthermore, they made the
surprising observation that when two such pulses are launched
towards each other, despite their nonlinear interaction, they
pass through each other without changing their shape.  Zabusky
and Kruskal coined the term \emph{soliton} for such solutions, in
recognition of their particle-like properties
\cite{zabusky1965isc}.

The spectacular properties of these solutions greatly excited the
mathematics and physics communities and many researchers started to study solitons.
In a celebrated paper \cite{gardner1967msk}, Gardner, Greene,
Kruskal and Miura uncovered some of the deep structure
underlying the KdV equation which is responsible for solitons and
their unusual properties.  The authors of \cite{gardner1967msk}
were studying the celebrated linear Schr\"odinger equation from
quantum mechanics, given by
\begin{IEEEeqnarray}{rCl}
\label{eq:lin-schr}
\psi_z(t,z) = \psi_{tt}(t,z) + g(t,z)\psi(t,z),
\end{IEEEeqnarray} 
where $\psi(t,z)$ is the wave-function
and $g(t,z)$ is an external potential.
They found that if one takes the solution $q(t,z)$ of the KdV
equation (\ref{eq:kdv}) as the external potential $g$ in
(\ref{eq:lin-schr}), then the eigenvalues of the Schr\"odinger
operator
\begin{equation}
	H=\frac{\partial^2}{\partial t^2}+g(t,z)
\label{eqn:sch-op}
\end{equation}
remain invariant during the
evolution in $z$!  Based on this critical observation, they
developed a method to recover the external potential $g(t,z) =
q(t,z)$ by solving an inverse problem for (\ref{eq:lin-schr}).
The method analytically predicts soliton solutions for the KdV
equation, as observed earlier by Zabusky and Kruskal through
numerical computations.  They had in fact found the IST for the
special case of the KdV equation.  

It was not immediately clear if the method developed in
\cite{gardner1967msk} could be generalized to other nonlinear
PDEs, since it is not obvious if there exists a certain auxiliary
operator, like the Schr\"odinger operator $H$, whose eigenvalues
are preserved during the evolution. In a landmark paper published
in 1968 \cite{lax1968ine}, Lax put the theory on a firm
mathematical footing. In particular, he established the
mathematical relationship between the auxiliary
operators with invariant eigenvalues (now called Lax pairs) and
the original nonlinear equation.  Once a Lax pair for a nonlinear
PDE is found, a method along the lines of \cite{gardner1967msk}
can be applied to solve that PDE.
  
Shortly afterwards, in 1972, Zakharov and Shabat found
a Lax pair for the NLS equation in one spatial dimension
\cite{zakharov1972ett},
and thus established that this equation,
too, could be solved in the same manner.
Details of this method for the NLS equation were
subsequently developed by Ablowitz and others
(see \cite{ablowitz2006sai} and references therein),
who also referred to this scheme as the ``nonlinear Fourier transform''.
After these discoveries from the 1960s and 70s,
research into solitons became an established
area of research, lying at the intersection of applied mathematics and nonlinear physics.

Nonlinear PDEs solvable by the NFT are called integrable
equations or exactly solvable models \cite{zakharov1991wi}.
These are usually Hamiltonian systems having an infinite number of conserved
quantities, and include the KdV, NLS, modified KdV,
sine-Gordon equations, and the Toda lattice, among
others  \cite{ablowitz2006sai,faddeev2007hmi}.
These equations all exhibit similar properties,
including the existence of soliton solutions.
The generation and processing of soliton signals, and the
engineering of novel transmission systems that support their
propagation, is described in \cite{singer1996spc}.  

\section{Canonical Lax Form for Exactly Solvable Models}
\label{sec:laxform}
\subsection{Lax Pairs and Evolution Equations}

We wish to consider linear differential operators
whose eigenvalues are invariant during an evolution \cite{lax1968ine}.
More precisely, we consider such operators
defined in terms of a signal $q(t,z)$
where the eigenvalues of the operator remain constant even
as $q$ evolves (in $z$) according to some evolution equation.

To facilitate the discussion, it is useful
to imagine a linear operator represented as a matrix; however,
we must keep in mind that,
 when moving from finite-dimensional spaces 
to infinite-dimensional spaces (of \eg, functions and operators),
some results do not carry over necessarily. 
The relevant properties of linear operators needed
for this paper are reviewed in Appendix~\ref{app:stso}.

Let $L(z)$ be a square matrix whose entries are functions of $z$.
Clearly, the eigenvalues of this matrix are in general functions of $z$ too.
However, for some matrices, it might be the case that while the entries of
the matrix change with $z$, the eigenvalues remain constant (independent of $z$).
Such a matrix, if diagonalizable, should be
similar to a constant diagonal matrix $\Lambda$, \ie,
$L(z)=G(z)\Lambda G^{-1}(z)$, for some similarity transformation
$G(z)$.

This idea generalizes to operators.
Let $\mathcal{H}$ be a Hilbert space,
let $\mathcal{D}$ be some domain that is dense in $\mathcal{H}$,
and let $L(z): \mathcal{D}\rightarrow \mathcal{H}$ be a family
of bounded linear operators indexed by a parameter $z$
\cite{reed1980mmm}.
If the eigenvalues of
$L(z)$ do not depend on $z$,
then we refer to $L(z)$ as an \emph{isospectral}
family of operators. If diagonalizable, it follows that for each $z$,
$L(z)$ is similar to a multiplication operator $\Lambda$
(the operator equivalent of a diagonal matrix; see Appendix~\ref{app:stso}),
\ie,
$L(z)=G(z)\Lambda G^{-1}(z)$,
for some operator $G(z)$.

Assuming that $L(z)$ varies smoothly with $z$, we can
consider the rate of the change (with
respect to $z$) of $L(z)$.  We have
\begin{IEEEeqnarray}{rCl}
\frac{\der L(z)}{\der z}&=&G_z\Lambda G^{-1}+G\Lambda
 \left(- G^{-1}G_z G^{-1} \right)\nonumber\\
&=&G_zG^{-1}\left(G\Lambda G^{-1}\right)-\left(G\Lambda
  G^{-1}\right)G_z
G^{-1}\nonumber \\
&=&M(z)L(z)-L(z)M(z)=\left[M,L\right],
\label{eq:dL-dt}
\end{IEEEeqnarray}
where
$G_z = \der G(z)/\der z$,
$M=G_z G^{-1}$, and $[M,L] \eqdef ML-LM$ is the \emph{commutator bracket}. In other words, every diagonalizable isospectral operator $L(z)$
satisfies the differential equation \eqref{eq:dL-dt}. 

Conversely, suppose $M(z)$ is given and the (unknown)
diagonalizable operator $L(z)$ evolves according to
\eqref{eq:dL-dt} with
initial condition $L(0)= G_0 \Lambda_0 G_0^{-1}$.
Let $G(z)$ be the (unique invertible) solution to $G_{z} = MG$ with $G(0) = G_0$.
One can easily verify that $L(z) = G(z) \Lambda_0 G(z)^{-1}$
satisfies \eqref{eq:dL-dt}.  Assuming that the solution to
a first-order differential equation is unique \cite{evanspde},
we see that $L(z)$ is an isospectral family.

The characterization of isospectral operators is therefore
summarized in the following lemma \cite{lax1968ine}.

\begin{lemma}
\label{lemma:isospectral}
Let $L(z)$ be a diagonalizable operator.  Then
$L(z)$ is isospectral if and only if it satisfies 
\begin{IEEEeqnarray}{rCl}
\frac{\der L}{\der z}=[M,L], 
\label{eq:lax}
\end{IEEEeqnarray}
for some operator $M$. If $L$ is self-adjoint 
(so that
$L$ is unitarily equivalent to a
multiplication operator, \ie, $L=G\Lambda G^*$),
then $M$ must be skew-Hermitian, \ie, $M^* = -M$.
\end{lemma}
\begin{proof}
The proof was outlined above.
The skew-Hermitian property of $M$ can be shown by differentiating
$G G^* = I$.
\end{proof}

\begin{figure}[t]
\centering
\ifCLASSOPTIONonecolumn
\includegraphics{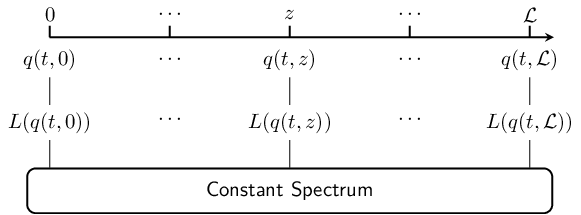}
\else
\includegraphics[scale=0.8]{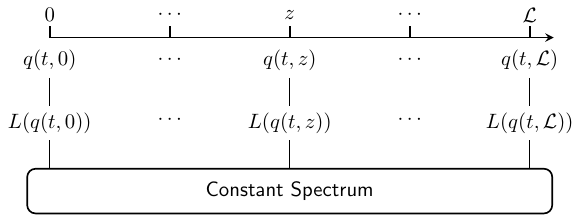}
\fi
\caption{An isospectral flow:  the spectrum of $L$
is held invariant even as $q(t,z)$ evolves.}
\label{fig:invariant}
\end{figure}

It is important to note that $L$ and $M$ do not have to be independent and
can depend on a common parameter, \eg, a function $q(t,z)$,
as illustrated in Fig.~\ref{fig:invariant}.
The isospectral
property of the solution is unchanged. The commutator bracket
$[M,L]$ in
\eqref{eq:lax} can create nonlinear evolution equations
for $q(t,z)$ in the form
\[
\frac{\partial q}{\partial z}=K(q),
\]
where $K(q)$ is some, in general nonlinear, function of $q(t,z)$ and its
time derivatives.   An example of this is the KdV equation.

\begin{example}[KdV equation]
\label{exam:kdv1}
Let $D=\frac{\partial}{\partial t}$
denote the time-derivative operator, and let
$q(t,z)$ be a real-valued function.
Finally, let
\begin{IEEEeqnarray*}{rCl}
L & =& D^2+q, \mbox{ and let}\nonumber \\
M & = & \frac{3}{2}D^3+\frac{1}{2}\left(Dq+qD\right).
\end{IEEEeqnarray*}
The Lax equation $L_z=[M,L]$ is easily simplified to 
\begin{IEEEeqnarray*}{rCl}
q_z - \left(q_{ttt}+qq_t\right)+(\textnormal{some terms})D \equiv \textbf{0},
\end{IEEEeqnarray*}
where $\textbf{0}$ is the zero operator.
The zero-order term of this equation as a polynomial in $D$, which must be zero, produces
the KdV equation
$ q_z=q_{ttt}+qq_t $.
Using $(D^{k})^*=(-1)^kD$ (provable by integration by parts), it is easy to see that $L$ and $M$ are self-adjoint and 
skew-Hermitian, respectively, since they are real-valued and
involve even and odd powers of $D$, respectively. The eigenvalues of
$L$ are thus preserved from Lemma~\ref{lemma:isospectral}.
Note that the $L$ operator in this example
is precisely the (linear) Schr\"odinger operator $H$ given
in \eqref{eqn:sch-op}.
\qed\end{example}
\begin{definition}
A pair of operators $L$ and $M$, depending on $z$, are called a Lax pair $(L,M)$ if they satisfy the Lax equation
\eqref{eq:lax}. Following Lemma~\ref{lemma:isospectral}, the eigenvalues of the $L$ operator are independent of
$z$.
\end{definition}

\subsection{The Zero-Curvature Condition}

The eigenvalues of the operator $L$, which
are constant in an isospectral flow, are defined via
\begin{equation}
Lv = \lambda v.
\label{eq:L-eq}
\end{equation}
Taking the $z$ derivative of \eqref{eq:L-eq} and using the Lax equation
$L_z=[M,L]$, we obtain $\left(L-\lambda I\right)\left(v_z-Mv\right)=0$.
Since $L-\lambda I$ vanishes only on eigenvectors of $L$,
it must be that $v_z - Mv$ is an element of the eigenspace, \ie,
$v_z - Mv = \alpha G x$, $\alpha\in \Complex$, and $x$ is
a coordinate vector. The
choice of $\alpha$ does not influence the results of future
sections; thus for simplicity we set $\alpha=0$.
It follows that an eigenvector $v(t,z)$ evolves based
on the linear equation 
\begin{equation}
\label{eq:M-eq}
v_z=Mv.
\end{equation}

Furthermore,
it is often more convenient to re-write \eqref{eq:L-eq} as
\begin{equation}
\label{eq:P-eq}
v_t= Pv,
\end{equation}
for some operator $P$. The relationship between $P$ and $L$
can be derived (if necessary)
by combining $\left(DI-P\right)v=0$ with $\left(L-\lambda I\right)v=0$,
obtaining
\begin{equation}
\label{eq:P-L}
P = \Sigma ( L - \lambda I) + DI,
\end{equation}
where $\Sigma$ is some invertible operator, and $D  = \frac{\partial}{\partial t}$ as in Example~\ref{exam:kdv1}.

Combining equations \eqref{eq:M-eq} and \eqref{eq:P-eq}
by using the equality
of mixed derivatives, \ie, $v_{tz}=v_{zt}$, the Lax equation
(\ref{eq:lax}) is reduced to the \emph{zero-curvature condition}
\cite{faddeev2007hmi}
\begin{equation}
\label{eq:zero-curv}
P_z-M_t+[P,M]=0.
\end{equation}
Note that the nonlinear equation derived
from (\ref{eq:zero-curv}) results as a compatibility
condition between the two linear equations
(\ref{eq:M-eq}) and (\ref{eq:P-eq}).  This shows
that certain nonlinear equations possess a
``hidden linearity'' in the form of (\ref{eq:M-eq}) and (\ref{eq:P-eq}).

Following the work of Zakharov and Shabat on the NLS
equation~\cite{zakharov1972ett}, Ablowitz et al.~\cite{ablowitz1974ist}
suggested
that for many equations of practical significance, the operator
$P$ can be fixed as
\begin{equation}
\label{eq:P-def}
P = \begin{pmatrix} -j\lambda & r(t,z) \\s(t,z) & j\lambda \end{pmatrix},
\end{equation}
where $r(t,z)$ and $s(t,z)$ are functions---depending
on $q(t,z)$---to be determined to
produce a given nonlinear evolution equation.
From \eqref{eq:P-L} this corresponds to the $L$ operator 
\begin{equation}
L = j \begin{pmatrix} D & -r(t,z) \\ s(t,z) & -D \end{pmatrix},
\label{eq:L-op}
\end{equation}
with $\Sigma=\diag(j,-j)$.
In this case, both $L$ and $P$ operate on $2 \times 1$ vector
functions.

Equation \eqref{eq:P-eq}, with $P$ as in \eqref{eq:P-def}, is known as the
\emph{AKNS system} (after the authors of \cite{ablowitz1974ist}) and
is central in the study of the nonlinear Fourier
transform \cite{faddeev2007hmi}. The important special case
where $r(t,z)=q(t,z)$ and $s(t,z)=-q^*(t,z)$ is generally known as
the \emph{Zakharov-Shabat system}. 
We will refer to (\ref{eq:L-eq}), (\ref{eq:M-eq}) and (\ref{eq:P-eq})
as the $L$-, $M$-, and $P$-equations, respectively.
Throughout this paper, we will assume that the domain of $L$, $P$, and
$M$ are subsets of the Hilbert space $L^2(\Reals)$, 
denoted here by $\mathcal{H}$, depending on the particular structure of that operator.

\begin{example}[Sine-Gordon equations]
Let $r(t,z)=-s(t,z)=\frac{1}{2}q_t(t,z)$ in \eqref{eq:P-def} and let
\begin{equation*}
M=\frac{j}{4\lambda}\begin{pmatrix} \cos(q) & -\sin(q) \\ -\sin(q) & -\cos(q) \end{pmatrix}.
\end{equation*}
Then the zero-curvature equation is simplified to $q_{tz}=\sin(q)$.
Taking $r=s=\frac{1}{2}q_t$ and $M$ as 
\begin{equation*}
M=\frac{j}{4\lambda}\begin{pmatrix} \cosh(q) & -\sinh(q) \\ \sinh(q) & -\cosh(q) \end{pmatrix}
\end{equation*}
gives $q_{tz}=\sinh(q)$.
\qed
\end{example}

\begin{example}[Nonlinear Schr\"odinger equation]
\label{exam:nls}
Take $r=q$, $s=-q^*$ and 
\[
M=\begin{pmatrix} 2j\lambda^2-j|q(t,z)|^2 & -2\lambda q(t,z)-jq_t(t,z) \\ 2\lambda q^*(t,z)-jq_t^*(t,z) & -2j\lambda^2+j|q(t,z)|^2\end{pmatrix}.
\]
The zero-curvature equation is simplified to
$jq_z(t,z)=q_{tt}(t,z)+2|q(t,z)|^2q(t,z)$.
\qed
\end{example}

\begin{example}[KdV equation revisited]
\label{ex:kdv2}
Let $r=\frac{2}{\sqrt{3}}q$, $s=-\frac{1}{4\sqrt{3}}$ and 
\begin{equation*}
M=\begin{pmatrix} 
4j\lambda^3-\frac{j\lambda q}{3} +\frac{q_t}{6} & 
-\frac{8\lambda^2q}{\sqrt{3}}-\frac{4}{\sqrt{3}}j\lambda q_t+\frac{2q^2}{3\sqrt{3}}+\frac{2q_{tt}}{\sqrt{3}}
\\  
\frac{\lambda^2}{\sqrt{3}}-\frac{q}{12\sqrt{3}}
& 
-4j\lambda^3+\frac{j\lambda q}{3} -\frac{q_t}{6} 
\end{pmatrix}.
\end{equation*}
The zero-curvature equation leads to $q_z=qq_t+q_{ttt}$. 
\qed
\end{example}

It should be noted from
Examples~\ref{exam:kdv1} and~\ref{ex:kdv2} that the
choice of $L$ and $M$ giving rise
to a given nonlinear equation is not unique.
Obviously one can scale $L$ by a number, or add a
constant $\alpha I$ to $L$ or $M$. In addition, both the Lax equation \eqref{eq:lax} and
\eqref{eq:L-eq} are unchanged under orthogonal transformations,
\ie, replacing $L$ and $M$ with $\Sigma
L\Sigma^T$ and $\Sigma M\Sigma^T$, respectively,
where $\Sigma$ is a (constant) orthogonal
matrix, \ie, $\Sigma^T\Sigma=I$.
Note further that it may be possible to
choose two Lax pairs $(L_1, M_1)$ and $(L_2,M_2)$ for a given equation
such that
$L_1$ is self-adjoint and $L_2$ is not self-adjoint. The eigenvalues
of $L_1$ and $L_2$ are, respectively, real and complex;  see
Appendix~\ref{app:stso}.

\subsection{Lax Convolution and Integrable Communication Channels}

Linear systems traditionally have been described by linear constant
coefficient differential equations.  An example is the one-dimensional heat
equation $q_z=c^2q_{tt}$, where $c$ is the diffusion coefficient and
$q(t,z)$ represents the heat profile across a rod extending in \emph{space} $t$, as \emph{time} $z$ goes on. From a systems point of view, this
defines a linear time-invariant system from input $x(t)=q(t,0)$ at
$z=0$ to the output $y(t)=q(t,\mathcal{L})$ at some $z=\mathcal{L}$. The role of $z$ is therefore just a
parameter and once fixed (to $z=\mathcal{L}$), the system is described by an
impulse response (or a Green function) $h(t;\mathcal{L})$, representing the
underlying (linear) convolution. 

Following this analogy, we wish to define a system in terms of a
Lax pair $(L,M)$.  Here, $L$ and $M$ are parametrized by a waveform
$q(t,z)$. Such a system accepts a waveform $x(t)=q(t,0)$ at its input and
produces a waveform $y(t)=q(t,\mathcal{L})$ at its output, according to the
evolution equation induced by $L_z = [ M, L ]$. The time-domain input-output
map is thus given by an evolution equation of the form $q_z=K(q)$,
obtainable from the Lax equation \eqref{eq:lax} (or its equivalent
\eqref{eq:zero-curv}). We refer to such a system as an
\emph{integrable system}. Note that an integrable system is completely
characterized by the two operators $(L,M)$ and the parameter $z=\const{L}$,
independent of the signals.  We denote such a system using
the triple $(L,M;\const{L})$.

\begin{figure}[t!]
\centering
\ifCLASSOPTIONonecolumn
\includegraphics{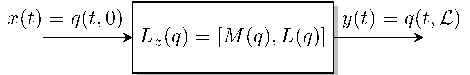}
\else
\includegraphics[scale=0.8]{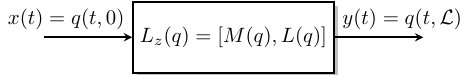}
\fi
\caption{A system defined by Lax convolution.} 
\label{fig:lax-conv}
\end{figure}

\begin{definition}[Lax convolution]  
We refer to the action of an integrable system $S=(L,M;\const{L})$
on the input $q(t,0)$ as the Lax convolution of $q$ with $S$.
We write the system output as
$q(t,\mathcal{L})=q(t,0)\convolution(L,M;\const{L})$. See Fig.~\ref{fig:lax-conv}.
\end{definition}

\begin{definition}[Integrable communication channels]
A waveform communication channel $C: x(t)\times n(t,z)\rightarrow y(t)$ with
inputs 
$x(t)\in L^1(\Reals)$ and space-time noise
$n(t,z)\in L^1(\Reals,\Reals^+)$, and output $y(t)\in L^1(\Reals)$, is
said to be \emph{integrable} if the noise-free channel is an integrable system.
\end{definition}

Note that noise can be introduced into an integrable system in a variety of
ways. In the approach taken in these papers, we assume that
the signal-to-noise ratio is not too
small, so that the stochastic system may justifiably be considered as a
perturbation of the deterministic system. Furthermore,
we limit ourselves to integrable channels with additive noise, \ie, 
\[
q_z=K(q)+n(t,z), \quad
q(t,0)=x(t),\quad
q(t,\const{L})=y(t), 
\]
where $n(t,z)$ is distributed band-limited noise. 
Note that this model is \emph{not} in general equivalent to
one in which noise is added (in lumped fashion) at the channel output.
Here, the noise is distributed in space, and so interacts with the
signal (in a potentially complicated manner) according to the
given evolution equation.  

In this paper, we develop a scheme for communication over integrable
channels.  By various choices of Lax pair $(L,M)$, one can construct a
variety of interesting channel models, mostly nonlinear, which go beyond the
linear channel models typically studied in data communications.
Interestingly, some linear channels can also be analyzed using this
nonlinear spectral approach advocated here \cite{Fokas97}.
As noted in Section~\ref{sec:introduction},
the central application (and motivation) for this work
is fiber-optic
communication, in which the channel model is given by the nonlinear
Schr\"odinger equation (\ref{eq:st-nls}), and
for which a Lax pair was given in Example~\ref{exam:nls}.
In this paper
we first discuss the deterministic (noise-free)
case and later, in [Part~III], treat
noise as a perturbation of the deterministic system.

\section{Nonlinear Fourier transform}
\label{sec:nft}
In this section, we assume that a function $q(t,\cdot)$
is given, and we define its nonlinear Fourier transform
with respect to the Lax operator $L$ \eqref{eq:L-op}.
As the notation $q(t,\cdot)$
implies, in this and the next section, the variable $z$ can take on any
value in the range $[0,\const{L}]$ and is irrelevant in the forward and
inverse transforms.  We shall therefore omit the index $z$ in what
follows.
\begin{assumptions}
We assume that $q(t)$ satisfies:
\begin{enumerate}
\item $q(t) \in L^1(\Reals)$;
\item $q(t) \rightarrow 0$ as $|t| \rightarrow \infty$.
\end{enumerate}
\label{assumptions:q}
\end{assumptions}

As previously noted, for concreteness we carry through
the development of the NFT for the case of the NLS equation (\ref{eq:nnls}),
for which the $P$-equation is the Zakharov-Shabat system
\begin{IEEEeqnarray}{rCl}
v_t&=&P(\lambda,q) v = \begin{pmatrix} -j\lambda & q(t) \\-q^*(t) & j\lambda \end{pmatrix}v .
\label{eq:dv-dt}
\end{IEEEeqnarray}

If for a given $\lambda \in \Complex$ the operator $L-\lambda I$
is not
invertible, then we say that
$\lambda$ belongs to the spectrum of $L$
and $v(t,\lambda)$ represents its associated eigenvector. In
finite-dimensional 
Hilbert spaces of matrices the spectrum is a discrete finite set of
eigenvalues. 
This may no longer be true in infinite-dimensional
spaces of operators, where the eigenvalues (if they exist)
may only be one part of the spectrum. See Appendix~\ref{app:stso}. 

The nonlinear Fourier transform of a signal $x(t)$ with
respect to an operator $L$ in a Lax pair is defined via the spectral analysis of
the $L$ operator, which we consider next.

\subsection{Canonical Eigenvectors and Spectral Coefficients}

We wish to study solutions of (\ref{eq:dv-dt}), in which
vectors $v(t) = (v_1(t), v_2(t))^T$ are considered as elements of the 
vector space $\mathcal{H}$.
We begin by equipping the vector space $\mathcal{H}$ with a symplectic
bilinear form $\mathcal{H}\times\mathcal{H}\mapsto\Complex$, which, for any \emph{fixed}
value of $t \in \mathbb{R}$, is defined as
\begin{equation*}
\inners{v(t)}{w(t)}=v_1(t) w_2(t) - v_2(t) w_1(t).
\end{equation*}
Let us also define the adjoint of any vector $v$ in $\mathcal{H}$
as
\[
\tilde{v}(t)=\left( \begin{array}{r}v_2^*(t)\\-v_1^*(t)\end{array}\right).
\]
The following properties hold true for all $v$ and $w$ in $\mathcal{H}$:
\begin{itemize}
\item $\tilde{\tilde{v}}=-v$;
\item $\inners{v}{v} = 0$;
\item $\inners{v}{w} = -\inners{w}{v}$;
\item $\inners{\tilde{v}}{v}=-\inners{v}{\tilde{v}}=|v_1|^2+|v_2|^2$;
\item $\frac{\der}{\der t}\inners{v}{w}=\inners{v_t}{w}+\inners{v}{w_t}$;
\item for every $2\times 2$ matrix $A$,
$\inners{Av}{w}+\inners{v}{Aw}=\tr (A)\inners{v}{w}$.
\end{itemize}

There are generally infinitely many solutions $v$ of \eqref{eq:dv-dt}
for a given $\lambda \in \Complex$, parametrized by the set of all possible boundary conditions.
These solutions form a subspace $E_{\lambda}$ of continuously
differentiable $2\times 1$ vector functions (an eigenspace).

\begin{lemma}
\label{lem:basis}
For all vectors $v(t)$ and $w(t)$ in $E_{\lambda}$,
\begin{enumerate}
\item $\tilde{v}\in E_{\lambda^*}$, \ie,
$\tilde{v}_t=P(\lambda^*,q)\tilde{v}$;\label{basis:prop2}
\item $\inners{v(t)}{w(t)}$ is a constant, independent of $t$;\label{basis:prop3}
\item If $\inners{v(t)}{w(t)}\neq 0$, then $v$ and $w$ are linearly
independent and form a basis for $E_{\lambda}$\label{basis:prop4};
\item $\dim(E_\lambda)=2$\label{basis:prop5}. 
\end{enumerate}
\end{lemma}

\emph{Proof}:
Property \ref{basis:prop2}) follows directly from (\ref{eq:dv-dt}).
To see \ref{basis:prop3}), note that
$\frac{\der }{\der t}\inners{v}{w} = \inners{v_t}{w} + \inners{v}{w_t}
= \inners{Pv}{w} + \inners{v}{Pw}
= \tr(P)\inners{v}{w}
= 0$. 
To see \ref{basis:prop4}), fix $t$ and
let $u(t)\in E_\lambda$, then $u(t)=a(t)v(t)+b(t)w(t)$ for some $a(t)$ and $b(t)$. Taking the symplectic inner product
of both sides with $w$ and $v$, we get $a(t)=\inners{u}{w}/\inners{v}{w}$ and
$b(t)=\inners{u}{v}/\inners{w}{v}$. From Property~\ref{basis:prop3}, $\inners{u}{w}$, $\inners{v}{w}$, $\inners{u}{v}$, and $\inners{w}{v}$ 
are all independent of $t$. It follows that $a$ and $b$ are also
independent of $t$. Finally, \ref{basis:prop5}) follows from \ref{basis:prop4}).
\qed

An important conclusion of Lemma~\ref{lem:basis} is that any two linearly
independent solutions $u$ and $w$ of \eqref{eq:dv-dt} provide a basis
for the solution space. To choose two such solutions,
we examine the behavior of the equation at large values of
$|t|$.
By Assumption~\ref{assumptions:q}.2,
as $|t|\rightarrow \infty$ \eqref{eq:dv-dt}
is reduced to 
\begin{IEEEeqnarray*}{rCl}
v_t& \rightarrow &\begin{pmatrix} -j\lambda & 0 \\0 & j\lambda \end{pmatrix}v,\qquad
\textnormal{ for large } |t|,
\end{IEEEeqnarray*}
which has a general solution
\begin{IEEEeqnarray*}{rCl}
v(t,\lambda) \rightarrow \left(\alpha e^{-j\lambda t}, \beta e^{j\lambda
  t}\right)^T, \quad \alpha, \beta \in \Complex.
\end{IEEEeqnarray*}

Two possible boundary conditions, bounded in the upper half
complex plane $\Complex^+ = \{ \lambda : \Im(\lambda)>0\}$, are
\begin{IEEEeqnarray}{rCl}
v^1(t,\lambda)&\rightarrow& \begin{pmatrix}0\\1\end{pmatrix}e^{j\lambda t},\:\:\:\quad
t\rightarrow +\infty \IEEEyesnumber\IEEEyessubnumber \label{eq:bound-cond-1},\\
v^2(t,\lambda)&\rightarrow& \begin{pmatrix}1\\0\end{pmatrix}e^{-j\lambda
  t}, \quad t\rightarrow -\infty.
\IEEEyessubnumber
\label{eq:bound-cond-2}
\end{IEEEeqnarray}
We solve (\ref{eq:dv-dt}) for a given $\lambda\in\Complex^+$
under the boundary conditions \eqref{eq:bound-cond-1}--\eqref{eq:bound-cond-2},
and denote the resulting solutions for all $t  \in \Reals$
as $v^1(t,\lambda)$ and $v^2(t,\lambda)$.
We can also solve
(\ref{eq:dv-dt}) for $\lambda^*$ under the adjoint boundary
conditions, bounded in the lower half complex plane $\Complex^{-}=\{\lambda : \Im(\lambda)<0\}$,
\begin{IEEEeqnarray}{rCCl}
v^1(t,\lambda^*)&\rightarrow& \begin{pmatrix}0\\1\end{pmatrix}e^{j\lambda^*t},\qquad
\:\: t\rightarrow +\infty, \IEEEyesnumber\IEEEyessubnumber\label{eq:adj-bound-cond-1}\\
v^2(t,\lambda^*)&\rightarrow& \begin{pmatrix}-1\\0\end{pmatrix}e^{-j\lambda^* t}, \quad t\rightarrow -\infty,
\IEEEyessubnumber\label{eq:adj-bound-cond-2}
\end{IEEEeqnarray}
giving rise to two solutions
$v^1(t,\lambda^*)$ and $v^2(t,\lambda^*)$.  Taking tilde ``$\sim$'' operation, we
obtain $\tilde v^1(t,\lambda^*)$ and $\tilde v^2(t,\lambda^*)$. From Lemma~\ref{lem:basis}
we have that $\tilde{v}^1(t,\lambda^*)$ and
$\tilde{v}^2(t,\lambda^*)$ are elements of $E_{\lambda}$.
These four eigenvectors
$v^1(t,\lambda)$,
$v^2(t,\lambda)$,
$\tilde{v}^1(t,\lambda^*)$, and
$\tilde{v}^2(t,\lambda^*)$, all of them elements of $E_{\lambda}$,
are called \emph{canonical eigenvectors}.
Fig.~\ref{fig:can-eigen} illustrates the canonical eigenvectors
at their boundaries.

\begin{figure}[t]
\centering
\ifCLASSOPTIONonecolumn
\includegraphics{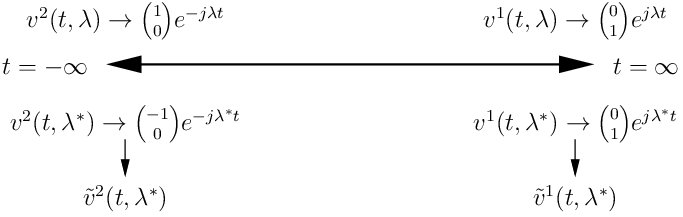}
\else
\includegraphics[width=\columnwidth]{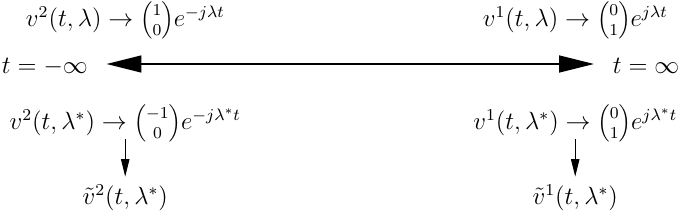}
\fi
\caption{Boundary conditions for the canonical eigenvectors.}
\label{fig:can-eigen}
\end{figure}

\begin{lemma}
Canonical eigenvectors satisfy:
\begin{enumerate}
  \item $\langle \tilde{v}^1(t,\lambda^*), v^1(t,\lambda\rangle_s=\langle v^2(t,\lambda),\tilde{v}^2(t,\lambda^*)\rangle_s=1$;
\item $\left\{v^1(t,\lambda),\tilde{v}^1(t,\lambda^*)\right\}$ and $\left\{v^2(t,\lambda),\tilde{v}^2(t,\lambda^*)\right\}$
are independent sets in $E_\lambda$;
\item $\widetilde{v^1(t,\lambda)}=\tilde{v}^1(t,\lambda^*)$ and  $\widetilde{v^2(t,\lambda)}=-\tilde{v}^2(t,\lambda^*)$. 
\end{enumerate}
\end{lemma}

\emph{Proof}:
1) Since $\inners{\tilde{v}^1}{v^1}$ is independent of $t$, using
\eqref{eq:bound-cond-1} and \eqref{eq:adj-bound-cond-1},
$\inners{\tilde{v}^1(t,\lambda^*)}{v^1(t,\lambda)}=
\inners{\tilde{v}^1(+\infty,\lambda^*)}{v^1(+\infty,\lambda)}=1$. 2)
Follows from 1) and Lemma~\ref{lem:basis}. 3) All four canonical eigenvectors
are in $E_\lambda$. This part thus follows by noting the 
boundary conditions \eqref{eq:bound-cond-1}--\eqref{eq:adj-bound-cond-2}.\qed

Choosing $\tilde{v}^1(t,\lambda^*)$ and $v^1(t,\lambda)$ as a basis of
$E_\lambda$, one can project $v^2(t,\lambda),\tilde{v}^2(t,\lambda^*)\in
E_\lambda$ on this basis to obtain
\begin{IEEEeqnarray}{rCl}
v^2(t,\lambda)&=&a(\lambda)\tilde{v}^1(t,\lambda^*)+b(\lambda)v^1(t,\lambda), \IEEEyesnumber\IEEEyessubnumber\label{eq:proj-1}\\
\tilde{v}^2(t,\lambda^*)&=&-b^*(\lambda^*)\tilde{v}^1(t,\lambda^*)+a^*(\lambda^*)v^1(t,\lambda), \IEEEyessubnumber\label{eq:proj-2}
\end{IEEEeqnarray}
where $a(\lambda)=\inners{v^2}{v^1}$,
$b(\lambda)=\inners{\tilde{v}^1}{v^2}$, and the coefficients in
\eqref{eq:proj-2} are obtained by applying the tilde ``$\sim$''
operation to \eqref{eq:proj-1}. A crucial property, following
from Lemma~\ref{lem:basis}, is that
$a(\lambda)$ and $b(\lambda)$
are time-independent.
The time-independent complex scalars $a(\lambda)$ and $b(\lambda)$ are
called the
\emph{nonlinear Fourier coefficients} \cite{ablowitz2006sai,faddeev2007hmi}.

Since the nonlinear Fourier coefficients are time independent, to
facilitate computing them, for simplicity we can send $t$ to, \eg,
$+\infty$ where $v^1(+\infty,\lambda)$ and $\tilde{v}^1(+\infty,\lambda^*)$
are known. The other two canonical eigenvectors  $v^2$ and $\tilde{v}^2$ are then
propagated from their boundary values $v^2(-\infty,\lambda)$ and
$\tilde{v}^2(-\infty,\lambda^*)$  at $t=-\infty $ according to
$v_t^2=P(\lambda,q)v^2$
and $\tilde{v}^2_t=P(\lambda^*,q)\tilde{v}^2$ to obtain $v^2(\infty,\lambda)$
and $\tilde{v}^2(\infty,\lambda^*)$.
At this stage, we have available all four canonical eigenvectors  at one time,
namely,
\[
\left\{v^2(+\infty,\lambda),
\tilde{v}^2(+\infty,\lambda^*),v^1(+\infty,\lambda),
\tilde{v}^1(+\infty,\lambda^*)\right\}.
\]
We can now project $v^2(+\infty,\lambda)$ and $\tilde{v}^2(+\infty,\lambda^*)$
onto the basis $v^1(+\infty,\lambda)$ and $\tilde{v}^1(+\infty,\lambda^*)$
according to \eqref{eq:proj-1}--\eqref{eq:proj-2} to obtain
 \begin{IEEEeqnarray}{rCl}
\left(v^2(+\infty,\lambda),
\tilde{v}^2(+\infty,\lambda^*)\right)=\left(
\tilde{v}^1(+\infty,\lambda^*),
v^1(+\infty,\lambda)\right) S,\nonumber\\
\label{eq:S-definition}
\end{IEEEeqnarray}
where
\begin{IEEEeqnarray*}{rCl}
S=\begin{pmatrix}
a(\lambda) & -b^*(\lambda^*)  \\
b(\lambda) & a^*(\lambda^*)
\end{pmatrix}.
\end{IEEEeqnarray*}

The matrix $S$ is called the scattering matrix and contains the nonlinear
Fourier coefficients \cite{ablowitz2006sai,faddeev2007hmi}. It is a function
of $q(t,\cdot)$ and says how the
solution to (\ref{eq:dv-dt}) is 
scattered from $t=-\infty$ to $t=+\infty$.   More precisely, the
field $v^2(-\infty,\lambda)=(1,0)^Te^{-j\lambda t}$ is applied
at $t=-\infty$, where $q$ is absent. This field
evolves forward in time according to (\ref{eq:dv-dt}),
interacts with the signal (which can be viewed as an
``obstacle'') at finite values of $t$, and subsequently propagates 
towards $t=+\infty$, where again $q$ is absent. The field at
$t=+\infty$ is measured and gives information about the ``obstacle''
as seen from a distance.  Although not obvious from the
development so far, we shall see in Section~\ref{sec:inft}
that the information measured at $t = +\infty$,
captured by $a(\lambda)$ and $b(\lambda)$,
is \emph{complete}, in the sense that from this information we
can retrieve $q(t,\cdot)$ entirely.
In view of this interpretation,
the nonlinear Fourier transform
was historically referred to as the \emph{inverse scattering transform}.

Note that, taking the determinant of the both sides of \eqref{eq:S-definition}, 
$\det S=a(\lambda)a^*(\lambda^*)+b(\lambda)b^*(\lambda^*)=1$.

\subsection{The Nonlinear Fourier Transform}

The projection equations \eqref{eq:proj-1} and \eqref{eq:proj-2}
that give $a(\lambda)$ and $b(\lambda)$
are well-defined if $\lambda\in\Reals$.
From Lemma~\ref{lem:basis}, Property 1), we
observe that the eigenspace is symmetric in 
$\lambda$, \ie, if $\lambda$ is an eigenvalue then
so is $\lambda^*$.
Thus it is sufficient to consider the upper half complex
plane $\Complex^+$. In this region,
the boundary conditions on the basis vectors
$v^1$ and $\tilde{v}^1$ at $t=\infty$ decay and blow up, respectively.
As a result, \eqref{eq:proj-1} is consistent in $\Complex^+$
only if $a(\lambda)=0$.
Eigenvalues in $\Complex^+$ are therefore identified as
the zeros of the complex function
$a(\lambda)$ and they form the discrete (point) spectrum of the signal.
We will see in the next section
that the discrete spectrum corresponds to soliton pulses. 

\begin{lemma}
\label{lem:a-analyticity}
If $q(t)\in L^1(\Reals)$,
 $a(\lambda)$ is an analytic function of $\lambda$ on
$\Complex^+$.
\end{lemma}
\emph{Proof}:
From Lemma~\ref{lem:analyticity} in Section~\ref{sec:inft},
if $q(t)\in L^1(\Reals)$ the scaled canonical eigenvectors 
$v^1(t,\lambda) e^{-j\lambda t}$ and $v^2(t,\lambda) e^{j\lambda t}$
are analytic functions of $\lambda$ in $\Complex^+$.
Since
\begin{equation*}
a(\lambda)=\inners{v^2}{v^1}=\inners{v^2e^{j\lambda t}}{v^1e^{-j\lambda t}},
\end{equation*}
is a combination
of two analytic functions in $\Complex^+$,
it is analytic in the same region. 
(Note, however, that $b(\lambda)$, which is
a combination of functions analytic in disjoint regions in $\Complex$,
may not be analytic in either of those regions.)
\qed

A consequence of Lemma~\ref{lem:a-analyticity}
is that the zeros of $a(\lambda)$ in $\Complex^+$ are isolated
points \cite{stein2003ca}.
It follows that the 
Zakharov-Shabat operator for the NLS equation has two
types of spectra. The discrete (or point) spectrum, which occurs in
$\Complex^+$, is characterized by those $\lambda_j \in \Complex^+$
satisfying
\[
a(\lambda_j)=0,~j=1,2,\ldots,N.
\]
The discrete spectrum corresponds to solitons, and
in this case (\ref{eq:proj-1}) reduces to
\[
v^2(t,\lambda_j)=b(\lambda_j)v^1(t,\lambda_j).
\]
The continuous spectrum,
which in general includes the whole
real line $\Im(\lambda)=0$, corresponds to the non-solitonic
(or radiation) component of the
signal. The continuous spectrum is the component
of the NFT which corresponds to the ordinary
Fourier transform, whereas the discrete spectrum
has no analogue in linear systems theory.
The reader is referred to
Appendix~\ref{app:stso} for a number of examples
illustrating various notions of the spectrum associated
with bounded linear operators.

To distinguish between the discrete and continuous
spectra, we find it convenient to refer to discrete
spectral values of $\lambda$ using the symbol $\lambda_j$ (with a subscript).
Continuous spectral values are denoted as $\lambda$ (without
a subscript).  In general, $\lambda_j \in \Complex^+$,
whereas $\lambda \in \mathbb{R}$.
  
For the purpose of developing the inverse transform,
we find it sufficient to work with the ratios
\[
\hat{q}(\lambda)=\frac{b(\lambda)}{a(\lambda)},
\quad
\tilde{q}(\lambda_j)=\frac{b(\lambda_j)}{a_{\lambda}(\lambda_j)},
\]
where $a_{\lambda}(\lambda_j)$ denotes the derivative $\left.\der a(\lambda)/\der \lambda\right|_{\lambda=\lambda_j}$.

We can now formally define the nonlinear
Fourier transform of a signal with respect to
the Lax operator $L$ \eqref{eq:L-op} as follows. 

\begin{definition}[Nonlinear Fourier transform
\cite{tao2006nfa,ablowitz2006sai}]
Let $q(t)$ be a sufficiently smooth function
satisfying Assumptions~\ref{assumptions:q}.
The nonlinear Fourier transform of $q(t)$ with respect to the Lax
operator $L$ \eqref{eq:L-op} consists of the continuous and
discrete spectral functions
$\hat{q}(\lambda):\,\Reals\mapsto \Complex$ and $\tilde{q}(\lambda_j):\,\Complex^+\mapsto\Complex$
where
\begin{IEEEeqnarray*}{rCl}
\hat{q}(\lambda)=\frac{b(\lambda)}{a(\lambda)},\quad
\tilde{q}(\lambda_j)=\frac{b(\lambda_j)}{a_{\lambda}(\lambda_j)}, \quad
j=1,2,\ldots,N,
\end{IEEEeqnarray*}
in which
$\lambda_j$ are the zeros of $a(\lambda)$.
Here, the spectral coefficients $a(\lambda)$ and $b(\lambda)$
are given by
\begin{IEEEeqnarray}{rCl}
a(\lambda)&=&\lim\limits_{t\rightarrow\infty}v^2_1e^{j\lambda t},
\IEEEyesnumber\IEEEyessubnumber\label{eq:a-formula}
\\
b(\lambda)&=&\lim\limits_{t\rightarrow\infty}v^2_2e^{-j\lambda t},
\IEEEyessubnumber\label{eq:b-formula}
\end{IEEEeqnarray}
where $v^2(t,\lambda)$ is a solution of (\ref{eq:dv-dt})
under the boundary condition (\ref{eq:bound-cond-2}).\qed
\end{definition}

To obtain the continuous spectral function $\hat{q}(\lambda)$,
$\lambda \in \mathbb{R}$, it is not
necessary to find $a(\lambda)$ and $b(\lambda)$ separately.
For convenience,
one can instead write an explicit differential equation
for $y(t,\lambda) = (v_2^2 / v_1^2 ) \exp(-2j \lambda t)$:
\begin{IEEEeqnarray}{Cl}
& \frac{\der y(t,\lambda)}{\der t}+q(t)e^{2j\lambda
  t}y^2(t,\lambda)+q^*(t)e^{-2j\lambda t}=0,
\label{eq:qhat-ode}
\\*[-0.5\normalbaselineskip]
\smash{\left\{
\IEEEstrut[7\jot]
\right.} \nonumber
\\*[-0.5\normalbaselineskip]
& 
y(-\infty,\lambda)=0,\nonumber
\end{IEEEeqnarray}
and obtain 
$\hat{q}(\lambda)=\lim_{t\rightarrow\infty}y(t,\lambda)$. 

Analogously, following \eqref{eq:a-formula}, one
can solve the second-order differential equation 
\begin{IEEEeqnarray}{Cl}
&
\frac{\der^2z(t,\lambda)}{\der
  t^2}-\left(2j\lambda+\frac{q_t}{q}\right)\frac{\der
  z(t,\lambda)}{\der t}+|q|^2z(t,\lambda) =0,\IEEEeqnarraynumspace
\label{eq:a-ode}\\*
[-0.2\normalbaselineskip]
\smash{\left\{
\IEEEstrut[8\jot]
\right.} \nonumber
\\*[-0.2\normalbaselineskip]
&
z(-\infty,\lambda)=1,\quad \frac{\der z(-\infty,\lambda)}{\der t}=0,\nonumber
\end{IEEEeqnarray}
and obtain $a(\lambda)=\lim\limits_{t\rightarrow\infty}z(t,\lambda)$. 
The discrete spectrum is obtained as
the zeros of $a(\lambda)$.

Just like the ordinary Fourier transform, the nonlinear Fourier
transform can be computed analytically
only in a few cases. An example is given in the next subsection. 

\subsection{Example:  Nonlinear Fourier Transform of a Rectangular Pulse}

Consider the rectangular pulse
\begin{IEEEeqnarray*}{rCl}
q(t)=\begin{cases}
A, & t \in [t_1,t_2], \\
0, & \textnormal{otherwise}.
\end{cases}
\end{IEEEeqnarray*}
Let $T=t_2-t_1$ and $T^\prime=t_2+t_1$.

In this case $P(\lambda,q)$ is time-independent when $t \in [t_1,t_2]$,
and \eqref{eq:dv-dt} under the boundary condition \eqref{eq:bound-cond-2}
can be easily solved in closed form.
The canonical eigenvector $v^2$ is given by
\[
v^2(t,\lambda) = \exp[(t-t_1)P] v^2(t_1,\lambda),
\quad v^2(t_1,\lambda) =
\begin{pmatrix} 1 \\ 0 \end{pmatrix} e^{-j \lambda t_1}.
\]
It follows that
\begin{IEEEeqnarray*}{rCl}
v^2(\infty,\lambda) =v^2(t_2,\lambda) =  \exp(PT) v^2(t_1,\lambda),
\end{IEEEeqnarray*}
where
\begin{IEEEeqnarray*}{ll}
\exp&(PT) = 
\exp\left\{
\begin{pmatrix}
-j\lambda & q\\
-q^* & j\lambda
\end{pmatrix}T
\right\}\\
&
=\begin{pmatrix}
\cos(\Delta T)-j\frac{\lambda}{\Delta}\sin(\Delta T) & \frac{A}{\Delta}\sin(\Delta T)\\
  \frac{-A^*}{\Delta}\sin(\Delta T)& 
\cos(\Delta T)+j\frac{\lambda}{\Delta}\sin(\Delta T)
\end{pmatrix},
\end{IEEEeqnarray*}
with $\Delta=\sqrt{\lambda^2+|A|^2}$. The spectral coefficients are
obtained from (\ref{eq:a-formula}) and (\ref{eq:b-formula}) as
\begin{IEEEeqnarray*}{rCl}
a(\lambda) & =& \left( \cos(\Delta T)-j\frac{\lambda}{\Delta}\sin(\Delta T)\right)e^{j\lambda T},\\
b(\lambda) &=& \frac{-A^*}{\Delta}\sin(\Delta T)e^{-j\lambda T^\prime}.
\end{IEEEeqnarray*}
The zeros of $a(\lambda)$ in $\Complex^+$, which satisfy
\begin{IEEEeqnarray*}{rCl}
j\tan(T\sqrt{|A|^2+\lambda^2})=\sqrt{1+\frac{|A|^2}{\lambda^2}},
\end{IEEEeqnarray*}
give rise to the discrete spectrum.
The continuous spectrum is given by
\begin{IEEEeqnarray*}{rCl}
	\hat{q}(\lambda)=\frac{A^*}{j\lambda}
e^{-2j\lambda t_2}\left( 1-\frac{\Delta}{j\lambda}\cot(\Delta T)\right)^{-1}.
\end{IEEEeqnarray*}

Note that as $A\rightarrow 0$, $\Delta\rightarrow \lambda$, and one can see that in the limit of $AT\ll 1$ there is no discrete spectrum.  Furthermore,
the continuous spectrum tends to
\begin{IEEEeqnarray*}{rCl}
	\hat{q}(\lambda)=-A^*Te^{-j2\pi f
	T^\prime}\sinc(2Tf), \quad \lambda=2\pi f,
\IEEEeqnarraynumspace
\end{IEEEeqnarray*}
which is just the ordinary Fourier transform of the $-q^*(t)$ with $f\rightarrow 2f$. 

\begin{figure*}[t]
\centering
\includegraphics[width=\textwidth]{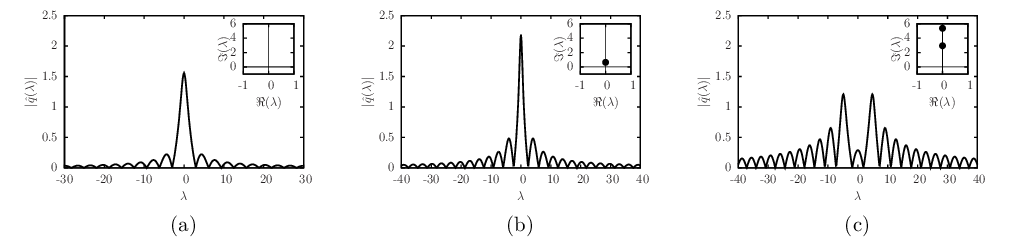}
\caption{Discrete and continuous
spectra of the square wave signal with $T=1$ and (a) $A=1$,
(b) $A=2$, (c) $A=6$.}
\label{fig:sw}
\end{figure*}

Fig.~\ref{fig:sw} 
shows the two spectra for $T=1$ and various values of $A$.
For small $A$, there is no discrete spectrum and the continuous
spectrum is essentially just the ordinary Fourier transform of $q(t)$.
As $A$ is increased, the continuous spectrum deviates from the ordinary Fourier transform and one or more discrete mass points appear on the 
$j\omega$ axis. 

\subsection{Elementary Properties of the Nonlinear Fourier Transform}
\label{subsec:elemprop}

Let $q(t)\leftrightarrow \left(\widehat{q}(\lambda),\widetilde{q}(\lambda_k)\right)$ be a nonlinear Fourier transform pair. The following properties
are proved in Appendix~\ref{sec:proof-prop-nft}:
\begin{enumerate}
\item (The ordinary Fourier transform as
limit of the nonlinear Fourier transform):  If $||q||_{L^1} \ll 1$, there is no discrete spectrum and $\widehat{q}(\lambda)\rightarrow Q(\lambda)$, where $Q(\lambda)$
is the ordinary (linear) Fourier transform of $-q^*(t)$
\begin{equation*}
Q(\lambda)=-\int\limits_{-\infty}^{\infty}q^*(t)e^{-2j\lambda t}\der t;
\end{equation*}

\item (Weak nonlinearity): If $|a|\ll 1$, then
$\widehat{aq}(\lambda) \approx a\widehat{q}(\lambda)$.
In general, however, $aq(t)\nleftrightarrow a(\hat q(\lambda),\tilde q(\lambda_k))$;

\item (Constant phase change): 
$e^{j\phi}q(t)\leftrightarrow e^{-j\phi}(\hat{q}(\lambda), \tilde{q}(\lambda_k))$;

\item (Time dilation): 
$q(\frac{t}{a})\leftrightarrow (\widehat{|a|q}(a\lambda),\widetilde{|a|q}(a\lambda_k))$;

\item (Time shift): $q(t-t_0)\leftrightarrow e^{-2j\lambda t_0}\left(\widehat{q}(\lambda), \widetilde{q}(\lambda_k)\right)$;

\item (Frequency shift): $q(t)e^{-2j\omega t}\leftrightarrow \left(\widehat{q}(\lambda-\omega),\widetilde{q}(\lambda_k-\omega)\right)$;

\item (Lax convolution): If $q_2(t)=q_1(t)\convolution(L,M;\const{L})$, then
$\widehat{q_2}(\lambda)=H(\lambda,\const{L})\widehat{q_1}(\lambda)$ and
$\widetilde{q_2}(\lambda_k)=H(\lambda_k,\const{L})\widetilde{q_1}(\lambda_k)$.
For the NLS equation, the channel
filter is $H(\lambda,\const{L})=\exp(-4j \lambda^2 \const{L})$;

\item (Parseval's identity): 
$\int\limits_{-\infty}^{\infty} \|q(t)\|^2 \, \der t = \hat{E} + \tilde{E}$, where
\[
\hat{E}  = 
\frac{1}{\pi}\int\limits_{-\infty}^{\infty}\log\left(1+|\hat{q}(\lambda)|^2\right)\der \lambda, \quad
\tilde{E}  = 
4\sum\limits_{j=1}^{N}\Im\left(\lambda_j\right).
\]
The quantities $\hat{E}$ and $\tilde{E}$ represent
the energy contained in the continuous and discrete spectra, respectively.
\end{enumerate}

In addition, we have the following properties related to the nonlinear Fourier
transform:
\begin{enumerate}
\item (Causality and layer-peeling property):
Let $q_1(t)$ and $q_2(t)$ be two signals with non-overlapping
support, \eg, signals in a pulse train. Without loss of
generality, assume that $q_1(t)$ is supported on
$t \leq t_0$, and that $q_2(t)$ is supported on $t > t_0$.
If $(a_1, b_1)$ and $(a_2, b_2)$ are,
respectively, the nonlinear Fourier coefficients of $q_1(t)$ and
$q_2(t)$, then the nonlinear Fourier coefficients
of $q_1(t)+q_2(t)$ are given by 
\begin{IEEEeqnarray*}{rCl}
(a(\lambda),b(\lambda)) & = & (a_1(\lambda),b_1(\lambda)) \circ (a_2(\lambda),b_2(\lambda)) \\
&\stackrel{\Delta}{=}& \bigl (a_1(\lambda)
a_2(\lambda)-b_1(\lambda)b_2^*(\lambda^*), \\
 & & \quad a_1(\lambda) b_2(\lambda)+b_1(\lambda)a_2^*(\lambda^*)\bigl).
\end{IEEEeqnarray*} 
 That is to say, if we slice the
 signal in time in
 consecutive portions according to a mesh $-\infty<\cdots<t_{-1}<t_0<t_1<\cdots<\infty$, the nonlinear Fourier coefficients
satisfy the Markov property
\begin{multline*}
\left(a_{k}(\lambda), b_{k}(\lambda); q(t), t\in(t_k,t_{k+1})
\right)\\
\rightarrow \left(a_{k+1}(\lambda), b_{k+1}(\lambda)\right),
\end{multline*}

where $\left(a_k(\lambda), b_k(\lambda)\right)$ are the spectral
coefficients calculated from $q(t)$ in $-\infty<t<t_{k}$;

\item (Purely imaginary eigenvalues) If $||q||_{L^1}<\frac{\pi}{2}$,
  then there is no discrete spectrum. If $q(t)\in L^1(\Reals)$ is real, non-negative,
  piece-wise smooth, and single-lobe (non-decreasing for $t<t_0$ and non-increasing for
  $t>t_0$), then there are precisely $N=
  \lfloor\frac{1}{2}+\frac{||q(t)||_{L^1}}{\pi}-\epsilon\rfloor$ eigenvalues,
  all purely imaginary and simple \cite{klaus2003ezss}. 
\end{enumerate}

\section{Evolution of the Nonlinear Fourier Transform} 

Derivation of the evolution of the nonlinear Fourier transform
of a signal propagating based on the NLS equation 
proceeds straightforwardly. As $q(t,z)$ propagates,
the eigenvalues of $L$ are preserved and the eigenvectors  of $L$
propagate based on 
\begin{IEEEeqnarray}{rCl}
v_z&=&M(\lambda,q)v \nonumber\\
& =& \setlength{\arraycolsep}{2pt}\begin{pmatrix} 2j\lambda^2-j|q(t,z)|^2 & -2\lambda q(t,z)-jq_t(t,z) \\
  2\lambda q^*(t,z)-jq_t^*(t,z) & -2j\lambda^2+j|q(t,z)|^2\end{pmatrix}v.
\IEEEeqnarraynumspace
\label{eq:dv-dz}
\end{IEEEeqnarray}

Assuming that $q(t,z)$ and its time-derivative vanish
at $t=\pm\infty$ for 
all $z\leq\const{L}$ during the propagation, then
as $t\rightarrow\infty$ \eqref{eq:dv-dz} is reduced to 
\begin{IEEEeqnarray}{rCl}
v_z(t,z)\rightarrow\begin{pmatrix} 2j\lambda^2 & 0 \\ 0 & -2j\lambda^2\end{pmatrix}v(t,z).
\label{eq:M-asymp}
\end{IEEEeqnarray}
Thus the boundary conditions \eqref{eq:bound-cond-1} 
and \eqref{eq:bound-cond-2} are transformed to  
\begin{IEEEeqnarray*}{rCl}
	v^1(t,\lambda)&\rightarrow& \begin{pmatrix}0\\1\end{pmatrix}e^{j\lambda t}e^{-2j\lambda^2 z},\quad
t\rightarrow +\infty, \\
v^2(t,\lambda)&\rightarrow& \begin{pmatrix}1\\0\end{pmatrix}e^{-j\lambda
	t}e^{2j\lambda^2 z}, \quad t\rightarrow -\infty.
\end{IEEEeqnarray*}
These transformed boundary conditions are not consistent with the boundary
conditions \eqref{eq:bound-cond-1} and \eqref{eq:bound-cond-2} used to define
the canonical eigenvectors, due to the additional factors
$e^{\pm 2j\lambda^2 z}$. As
a result, the evolution of the canonical eigenvectors from $z=0$, according to
$v_z=Mv$, does not lead to the canonical eigenvectors at $z=\const{L}$. However, by
proper scaling, one can obtain the canonical eigenvectors at any $z$.

For instance, focusing on $v^2(t,\lambda;z)$, and changing
variables to
$u^2(t,\lambda;z)=v^2(t,\lambda;z)e^{-2j\lambda^2 z}$, we
obtain $u^2_t=Pu^2$ with boundary condition \eqref{eq:bound-cond-2} (at $t=-\infty$) for all
$z$. Consequently, $u^2(t,\lambda;z)$ is a canonical eigenvector for all $z$.
By transforming
\eqref{eq:M-asymp}, the evolution equation for $u$ is asymptotically
(at $t = +\infty$) given by
\begin{equation*}
u^2_z(\infty,\lambda;z)=
\begin{pmatrix}
0 & 0\\
0 & -4j\lambda^2 z
\end{pmatrix}u^2(\infty,\lambda;z),
\end{equation*}
which gives
\begin{IEEEeqnarray*}{rCl}
u_1^2(\infty,\lambda;z) & = & u_1^2(\infty,\lambda;0),\\
u_2^2(\infty,\lambda;z) & = & u_2^2(\infty,\lambda;0)e^{-4j\lambda^2 z}.
\end{IEEEeqnarray*}
Using expressions (\ref{eq:a-formula}) and (\ref{eq:b-formula}) we
obtain
\begin{IEEEeqnarray*}{rCl}
	a(\lambda,z)&=&\lim\limits_{t\rightarrow\infty}u_1^2(t,\lambda;z)e^{j\lambda t}\!=\!\lim\limits_{t\rightarrow\infty}u_1^2(t,\lambda;0)e^{j\lambda t}=a(\lambda,0),
\\
b(\lambda,z)&=&\lim\limits_{t\rightarrow\infty}u_2^2(t,\lambda;z)e^{-j\lambda t}\!=\!\lim\limits_{t\rightarrow\infty}u_2^2(t,\lambda;0)e^{-4j\lambda^2z}e^{-j\lambda t}
\\
&=&b(\lambda,0)e^{-4j\lambda^2 z}.
\end{IEEEeqnarray*}
In turn, the nonlinear Fourier transform propagates according to
\begin{IEEEeqnarray}{rCl}
\widehat{q(t,z)}(\lambda)&=&e^{-4j\lambda^2 z}\widehat{q(t,0)}(\lambda),\nonumber\\
\widetilde{q(t,z)}(\lambda_j)&=&e^{-4j\lambda_j^2 z}\widetilde{q(t,0)}(\lambda_j),\nonumber\\
\lambda_j(z)&=&\lambda_j(0),\qquad j=1,2,\ldots,N.
\label{eq:nft-evol}
\end{IEEEeqnarray}
Note that, since $a(\lambda,z)$ is preserved under the evolution (\ie,
is independent of $z$), the number of the discrete eigenvalues---which are
zeros of $a(\lambda)$---is also preserved.

In summary,
we see that the operation of the Lax convolution in
the nonlinear Fourier domain is described by a simple
multiplication (diagonal) operator, much in
the same way that the ordinary Fourier transform
maps $y(t)=x(t) \convolution h(t) $ to $Y(\omega)=X(\omega) \cdot
H(\omega)$. The channel filter (transfer function)
in \eqref{eq:nft-evol} is $\exp(-4j\lambda^2z)$.

\section{An Approach to Communication over Integrable Channels}

Since the nonlinear Fourier transform of a signal is
essentially preserved under Lax convolution, one can immediately
conceive of a nonlinear analogue of orthogonal frequency-division
multiplexing (OFDM) for communication over integrable channels.
We refer to this scheme as nonlinear frequency-division multiplexing. In this scheme, the input-output channel model is given by
\begin{IEEEeqnarray*}{rCl}
\hat{Y}(\lambda)&=&H(\lambda)\hat{X}(\lambda)+\hat{Z},\\
\tilde{Y}(\lambda_j)&=&H(\lambda_j)\tilde{X}(\lambda_j)+\tilde{Z}_j,
\end{IEEEeqnarray*}
where $\hat{X}(\lambda)=\hat{q}(\lambda,0)$ and
$\tilde{X}(\lambda_j)=\tilde{q}(\lambda_j,0)$ are spectra at the input of the
channel, $\hat{Y}(\lambda)=\hat{q}(\lambda,z)$ and
$\tilde{Y}(\lambda_j)=\tilde{q}(\lambda_j,z)$ are spectra at the output of the
channel,
and the channel filter is \begin{IEEEeqnarray*}{rCl}
H(\lambda)=e^{-4j\lambda^2 z}.
\end{IEEEeqnarray*} 
Here $\hat{Z}$ and $\tilde{Z}_j$ are effective noise in the spectral domain.
Bandlimited Gaussian noise in the time domain generally transforms to
non-Gaussian noise in the spectral domain. The noise variables affecting distinct eigenvalues are in general
correlated and a function of the entire signal spectrum. Exact
statistical description of the noise in the spectral domain is in general involved;
however when noise is small a perturbation approach can easily be taken
[Part~III].

The proposed scheme consists of two steps.
\begin{itemize}
\item \emph{The inverse nonlinear Fourier transform at the transmitter
    (INFT).} At the transmitter, information is encoded in the nonlinear
spectra of the signal according to a suitable constellation on $\left(\hat{X}(\lambda),\tilde{X}(\lambda_j)\right)$. The time domain signal is generated by taking the inverse
nonlinear Fourier transform, 
\[
q(t)=\mathrm{INFT}\left(\hat{X}(\lambda),\tilde{X}(\lambda)\right),
\]
and is sent over the channel.  (The INFT is described
formally in the next section.)
\item \emph{The forward nonlinear Fourier transform at the receiver (NFT).}
At the receiver, the (forward) nonlinear Fourier transform of the
signal,
\[
\left(\hat{Y}(\lambda),\tilde{Y}(\lambda)\right)=\mathrm{NFT}(q(t,z))
\]
is taken and the resulting spectra are compared against the
transmitted spectra according to some metric $d\left(\hat{X}(\lambda),\tilde{X}(\lambda);\hat{Y}(\lambda),\tilde{Y}(\lambda)\right)$. 
\end{itemize}

As $q(t,0)$ propagates in a communication network
modeled by a complicated integrable equation, it is significantly distorted
and undergoes inter-symbol interference (ISI) and
inter-channel interference. In the spectral domain, in the absence of
noise, all the nonlinear spectral components propagate independent of each other and
the channel is decomposed into a number of linear parallel independent
channels. By diagonalizing the channel
in this way, the deterministic ISI and inter-channel interference are removed in the spectral domain. 

In this scheme, as in linear OFDM,
communication objectives, such as constellation design, coding and
modulation are entirely formulated in the spectral
domain. All available degrees-of-freedom, \ie,
$\{\lambda_j,\tilde{X}(\lambda_j),\hat{X}(\lambda)\}$ can be generally
modulated.  Time domain constraints can be translated to constraints in
the spectral domain. A power constraint, for instance, can be exactly
transformed to a power constraint in the spectral domain with the help
of the Parseval's identity.
A bandwidth limitation is not directly and simply expressed
in the nonlinear spectral domain;  however, it appears that
the nonlinear spectrum of a signal bandlimited to $W$ is indeed confined,
approximately, to a vertical strip in the real line and the upper half
complex plane, with a width depending on $W$.

\begin{figure}[tbp]
\centering
\ifCLASSOPTIONonecolumn
\includegraphics{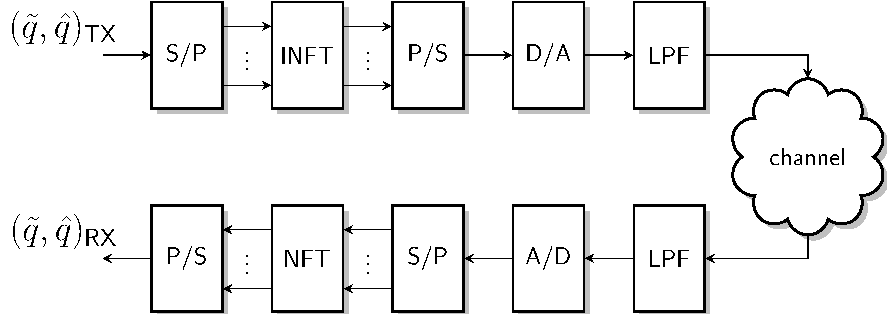}
\else
\includegraphics[width=\columnwidth]{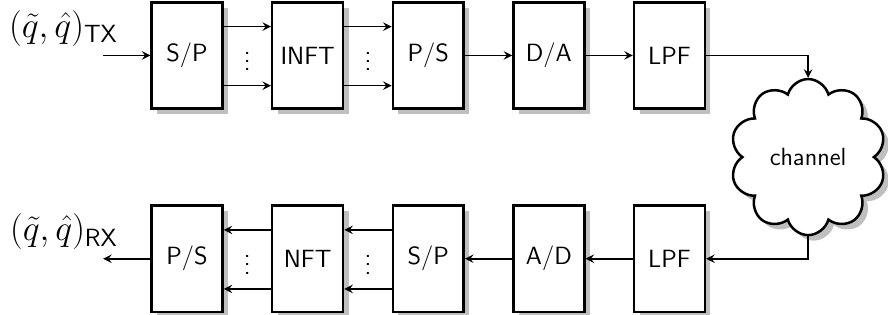}
\fi
\caption{Nonlinear frequency-division multiplexing system architecture.} 
\label{fig:nfdm-model}
\end{figure}

Fig.~\ref{fig:nfdm-model} illustrates the NFDM
channel model that we use in this paper (and its subsequent parts) for
data transmission over integrable channels such as optical
fibers. Although in this paper we introduced the continuous-time NFT
on the real line $t\in\Reals$, for practical implementation
the nonlinear Fourier transform of discrete-time signals with
periodic boundary conditions should be used; see [Part~III] and
references therein.

\section{Inverse Nonlinear Fourier transform}
\label{sec:inft}

In this section, we describe the inverse nonlinear Fourier
transform, which is a method for recovering the signal
$q(t,\cdot)$ from its nonlinear Fourier transform $(\hat{q}(\lambda),
\tilde{q}(\lambda_j))$.
This is the opposite of what was done in the
Section~\ref{sec:nft}, and it gives further important
insight into how the NFT works. As in
Section~\ref{sec:nft}, the value of $z$ is irrelevant here and the
index $z$ is thus suppressed.

\subsection{Riemann-Hilbert Factorization}

In the 1950s, before the publication of Gardner
\cite{gardner1967msk},
a method for retrieving the potential
$q(t)$ in the linear Schr\"odinger equation from the
knowledge of the scattering matrix $S$ was already known in
quantum mechanics.
It was later realized
that this method is an instance of a
\emph{Riemann-Hilbert problem} in complex analysis
\cite{tao2006nfa,faddeev2007hmi}.
Alternatively, the same task can be accomplished
using the \emph{Gelfand-Levitan method} \cite{faddeev2007hmi},
a different approach developed earlier in the
context of inverse problems for Sturm-Liouville differential equations.
Although
either of these methods can be used for solving the inverse problem,
in this paper we will use the Riemann-Hilbert method. 

To begin, canonical eigenvectors can be found in two ways. On the one
hand, they are
related to $q(t)$ through the Zakharov-Shabat system \eqref{eq:dv-dt}.
One can solve
$\eqref{eq:dv-dt}$ to explicitly express
the canonical eigenvectors as a series involving
$q(t)$. On the other hand, canonical eigenvectors are related to the
nonlinear Fourier transform via the projection equations
\eqref{eq:proj-1}--\eqref{eq:proj-2}. The latter are two equations for four
unknowns and in general cannot be solved.
However, the unique properties
of the canonical eigenvectors, which are analytic in disjoint regions of the complex
plane, will allow us to find them in terms of the nonlinear
Fourier transform. Note that this second
derivation does not explicitly depend on $q(t)$; however,
by equating the canonical eigenvectors obtained
from these two derivations,
we can relate the signal $q(t)$ to its nonlinear
Fourier transform ($\hat{q}(\lambda),\tilde{q}(\lambda_j))$.

We find it helpful to briefly introduce the tool that we employ in this
section, namely the Riemann-Hilbert factorization problem in
complex analysis.

\begin{definition}[Riemann-Hilbert factorization]
The Riemann-Hilbert factorization problem consists of
finding two complex functions
$f^{-}(z)$ and $f^{+}(z)$, respectively, analytic inside and outside of a
closed contour $C$ in the complex plane, such that on $C$ they satisfy
the boundary condition $f^{+}(z)=g(z)f^-(z)$ for some given $g(z)$.
\end{definition}

It $g(z)$ satisfies a H\"older condition, it is possible to find both $f^{+}(z)$ and $f^{-}(z)$ everywhere in
the complex plane from the given boundary condition. In
Appendix~\ref{app:rhp}, we provide a brief overview of 
this problem and its solution.  See
\cite{musk2008} for more discussion. 

\subsection{The Inverse Transform}

The inverse transform maps the two spectral functions 
$\left(\hat{q}(\lambda),\tilde{q}(\lambda_j)\right)$ to the signal
$q(t)$.

As noted,
we need to
express the canonical eigenvectors in terms of the
nonlinear Fourier transform. Below we will find
it more convenient to work with canonical eigenvectors subject to fixed
boundary conditions. Scaling the canonical eigenvectors as
\begin{IEEEeqnarray}{rCl}
\quad V^1 &=& v^1e^{-j\lambda t},
\quad \tilde{V}^1(t,\lambda^*)=\tilde{v}^1(t,\lambda^*)e^{j\lambda t},
\nonumber
\\
V^2&=&v^2e^{j\lambda t}, \quad \:\:\,
\tilde{V}^2(t,\lambda^*)=\tilde{v}^2(t,\lambda^*)e^{-j\lambda t}, 
\label{eq:scaled-eigvecs}
\end{IEEEeqnarray}
the projection equations \eqref{eq:proj-1} and \eqref{eq:proj-2} are
transformed to
\begin{IEEEeqnarray}{rCl} 
V^2(t,\lambda)&=&a(\lambda)\tilde{V^1}(t,\lambda^*)+b(\lambda)e^{2j\lambda
  t}V^1(t,\lambda) 
\IEEEyesnumber\IEEEyessubnumber\label{eq:MN-rel-1}, \\
\tilde{V^2}(t,\lambda^*)&=&-b^*(\lambda^*)e^{-2j\lambda
  t}\tilde{V^1}(t,\lambda^*)+a^*(\lambda^*)V^1(t,\lambda).
\IEEEyessubnumber
\IEEEeqnarraynumspace
\label{eq:MN-rel-2}
\end{IEEEeqnarray}  

\begin{lemma}
\label{lem:analyticity}
If $q\in L^1(\Reals)$, then $V^1(t,\lambda)$ and $V^2(t,\lambda)$ are analytic functions of $\lambda$ in the
upper half complex plane $\Complex^+$ while $\tilde{V}^1(t,\lambda^*)$ and $\tilde{V}^2(t,\lambda^*)$ are analytic functions of $\lambda$ in 
the lower half complex plane.
\end{lemma}

\emph{Proof:} See Appendix~\ref{sec:analyticity}.\qed

Rearranging \eqref{eq:MN-rel-1}--\eqref{eq:MN-rel-2} gives
\begin{multline}
\setlength{\arraycolsep}{2pt}
\left( \begin{array}{cc} V^1(t,\lambda) & V^2(t,\lambda) \end{array}\right)
 =  \left(
\setlength{\arraycolsep}{2pt}
\begin{array}{cc} \tilde{V^1}(t,\lambda^*) & \tilde{V^2}(t,\lambda^*)\end{array} \right)  \\
\times \:\left(
\setlength{\arraycolsep}{2pt}
 \begin{array}{cc} \frac{b^*(\lambda^*)}{a^*(\lambda^*)}e^{-2j\lambda t} & \frac{1}{a^*(\lambda^*)} \\
 \frac{1}{a^*(\lambda^*)} & \frac{b(\lambda)}{a^*(\lambda^*)}e^{2j\lambda t}\end{array} \right).
\label{eq:mrh} 
\end{multline}  
Since $\setlength{\arraycolsep}{2pt}\begin{pmatrix}V^1(t,\lambda) & V^2(t,\lambda)\end{pmatrix}$ and $\setlength{\arraycolsep}{2pt}\begin{pmatrix}\tilde{V^1}(t,\lambda^*)&
  \tilde{V^2}(t,\lambda^*)\end{pmatrix}$ are analytic, respectively, in the upper and lower half complex plane, \eqref{eq:mrh} defines a matrix Riemann-Hilbert problem.

Solution of the Riemann-Hilbert factorization problem \eqref{eq:mrh}
is given in Appendix~\ref{sec:sol-of-rh}. Omitting the details,
the following linear system of equations, known as the Riemann-Hilbert
system, is obtained:
\begingroup
\allowdisplaybreaks
\begin{IEEEeqnarray}{rCl}
\tilde{V^1}(t,\lambda^*)&=&\begin{pmatrix}1\\0\end{pmatrix}+\sum\limits_{i=1}^{N}\frac{\tilde{q}(\lambda_i)e^{2j\lambda_i
    t}V^1(t,\lambda_i)}{\lambda-\lambda_i}\IEEEnonumber\\
&&+\:
\frac{1}{2\pi j}\int\limits_{-\infty}^{\infty}\frac{\hat{q}(\zeta)e^{2j\zeta
    t}V^1(t,\zeta)}{\zeta-(\lambda-j\epsilon)}\der \zeta ,
\IEEEyesnumber\IEEEyessubnumber 
\label{eq:rh-integ-eqs-a}
\\
V^1(t,\lambda)&=&\begin{pmatrix}0\\1\end{pmatrix}-\sum\limits_{i=1}^{N}\frac{\tilde{q}^*(\lambda_i^*)e^{-2j\lambda_i^*
    t}\tilde{V^1}(t,\lambda_i^*)}{\lambda-\lambda_i^*}\IEEEnonumber\\
&&+\:
\frac{1}{2\pi j}\int\limits_{-\infty}^{\infty}\frac{\hat{q}^*(\zeta)e^{-2j\zeta
    t}\tilde{V^1}(t,\zeta^*)}{\zeta-(\lambda+j\epsilon)}\der
\zeta, \IEEEeqnarraynumspace
\IEEEyessubnumber
\label{eq:rh-integ-eqs-b}
\\
\tilde{V^1}(t,\lambda_m^*)&=&\begin{pmatrix}1\\0\end{pmatrix}+\sum\limits_{i=1}^{N}\frac{\tilde{q}(\lambda_i)e^{2j\lambda_i
    t}V^1(t,\lambda_i)}{\lambda_m^*-\lambda_i}\IEEEnonumber\\
&&+\:
\frac{1}{2\pi j}\int\limits_{-\infty}^{\infty}\frac{\hat{q}(\zeta)e^{2j\zeta
    t}V^1(t,\zeta)}{\zeta-\lambda_m^*}\der \zeta, \IEEEyessubnumber 
\label{eq:rh-integ-eqs-c}
\\
V^1(t,\lambda_m)&=&\begin{pmatrix}0\\1\end{pmatrix}-\sum\limits_{i=1}^{N}\frac{\tilde{q}^*(\lambda_i^*)e^{-2j\lambda_i^*
    t}\tilde{V^1}(t,\lambda_i^*)}{\lambda_m-\lambda_i^*}\IEEEnonumber\\
&&+\:
\frac{1}{2\pi j}\int\limits_{-\infty}^{\infty}\frac{\hat{q}^*(\zeta)e^{-2j\zeta
    t}\tilde{V^1}(t,\zeta*)}{\zeta-\lambda_m}\der
\zeta , \IEEEyessubnumber
\IEEEeqnarraynumspace
\label{eq:rh-integ-eqs-d}
\end{IEEEeqnarray}
\endgroup
where $\lambda,\zeta\in\Reals$, $\lambda_m\in\Complex^+$ and  $m=1,\ldots, N$. This is a system of $2N+2$ linear equations
for $2N+2$ discrete and continuous canonical eigenvectors
$\{V^1(t,\lambda_m)\}_{m=1}^{N}$, $\{\tilde{V^1}(t,\lambda_m)\}_{m=1}^{N}$,
and $V^1(t,\lambda)$ and $\tilde{V^1}(t,\lambda)$ as a function of the Fourier
transform $(\hat{q}(\lambda),\tilde{q}(\lambda_j))$.

Next, the canonical eigenvectors are related to the signal $q(t)$. By
inspecting the Zakharov-Shabat system \eqref{eq:dv-dt}, it is
shown in Appendix~\ref{app:asymp-eigenvectors} that for $|\lambda|\gg 1$  
\begin{IEEEeqnarray}{rCl} 
V^1(t,\lambda)=\begin{pmatrix}
\frac{1}{2j\lambda}q(t)\\
  1+\frac{1}{2j\lambda}\int\limits_{t}^{\infty}|q(t)|^2 \der t 
\end{pmatrix}+O(\lambda^{-2}).
\label{eq:N-asym0}
\end{IEEEeqnarray}  

It now becomes easy to recover the signal $q(t)$ from the nonlinear
Fourier transform
$(\hat{q}(\lambda),\tilde{q}(\lambda_j))$. Eigenvector $V^1$ is related to $q$ via \eqref{eq:N-asym0} for
$|\lambda|\gg 1$, and is related to the nonlinear Fourier transform
via \eqref{eq:rh-integ-eqs-b}. Approximating $1/(\lambda-\zeta)$ and
$1/(\lambda-\lambda_i)$ by $1/\lambda+O(\lambda^{-2})$ in \eqref{eq:rh-integ-eqs-b} for
$|\lambda|\gg 1$ and comparing the resulting $V^1_1$ with $V^1_1$ in \eqref{eq:N-asym0}, we obtain
\begin{IEEEeqnarray}{rCl}
q^*(t)&=&2j\sum\limits_{i=1}^{N}
\tilde{q}(\lambda_i)e^{2j\lambda_i t}V^1_{2}(t,\lambda_i)\nonumber\\
&&-\:
\frac{1}{\pi}\int\limits_{-\infty}^{\infty}\hat{q}(\lambda)e^{2j\lambda t}V^1_{2}(t,\lambda)\der \lambda.
\label{eq:potential-c}
\end{IEEEeqnarray}
This represents $q(t)$ as a function of the nonlinear Fourier transform.

In summary,
given 
$(\hat{q}(\lambda),\tilde{q}(\lambda_j))$, we
first solve the Riemann-Hilbert system
to find the eigenvector $V^1$. For this purpose, one can discretize
the system 
\eqref{eq:rh-integ-eqs-a}--\eqref{eq:rh-integ-eqs-d} and solve a linear system of equations
of the form
\begin{IEEEeqnarray*}{rCl}
Ax=b,
\end{IEEEeqnarray*}
for appropriate $A$ and $b$. This is done for each fixed $t_i$ to
determine values of $V^l(t_i,\lambda)$ at all times. Then,
$V^1_2$, $\hat{q}$ and $\tilde{q}$ are substituted in
\eqref{eq:potential-c} to obtain the signal $q(t)$.

Note that this inverse transform is taken only once at the transmitter to
synthesize the desired pulse shapes. It is only the forward transform which
needs to be taken in real time at the receiver. 

\begin{remark}
  In this paper we fixed $L$ to be the AKNS operator \eqref{eq:P-def}. In general, however $L$ needs
  to have adequate analyticity properties so that the inverse NFT can be carried out according
  to a Riemann-Hilbert factorization. Thus the existence of a Lax pair alone is not sufficient for the NFT. 
\end{remark}

The mathematical tools developed in this paper are used in
Parts II and III
of this paper to present further details pertaining the application of the suggested scheme in
data transmission over integrable channels. 

In [Part~II] numerical methods are suggested to compute
the nonlinear Fourier transform of a signal with respect
to the AKNS Lax operator. We use the results of Parts~I and II
in [Part~III] to study the NFT in the presence
of the noise and illustrate the performance of the method
in actual fiber-optic systems.

\section{Conclusions}

The nonlinear Fourier transform of a signal with respect to an operator $L$ in
a Lax pair consists of continuous and discrete spectral functions
$\hat{x}(\lambda)$ and $\tilde{x}(\lambda_j)$, obtainable by solving the
eigenproblem for the $L$ operator. The NFT maps a Lax convolution to a
multiplication operator in the spectral domain. Using the nonlinear Fourier
transform, we propose a transmission scheme for integrable channels, termed
nonlinear frequency-division multiplexing, in which the information in encoded
in the nonlinear spectrum of the signal. The scheme is an extension of
conventional OFDM to integrable channels generated by Lax pair operators.  An example is the
optical fiber channel, in which signals propagate according to the nonlinear
Schr\"odinger equation. The class of integrable channels, though nonlinear and
often complicated, are somehow ``linear in disguise,'' and thus admit the
proposed nonlinear frequency-division multiplexing transmission scheme.


\appendices

\section{Spectrum of Bounded Linear Operators}
\label{app:stso}

When moving from finite-dimensional spaces (of \eg, matrices) to
infinite-dimensional spaces (of \eg, functions and operators), some results do
not carry over necessarily.  Here we recall a few useful results in functional
analysis \cite{reed1980mmm}. 

Let $\mathcal{H}$ be a Hilbert space, let $\mathcal{D}$ be a dense subset of
$\mathcal{H}$, and let $L: \mathcal{D}\rightarrow \mathcal{H}$ be an operator.

\begin{definition}
The \emph{adjoint} of $L$ is the operator $L^*$ whose
domain $\mathcal{D}^*$ consists of all $\psi \in \mathcal{H}$ for which there
exists an element $L^*\psi \in \mathcal{H}$ such that
\[
\inner{L^*\psi}{\phi} =\inner {\psi}{L\phi} ,\quad \forall \phi\in \mathcal{D}.
\]
\end{definition}

The operator $L$ is said to be \emph{self-adjoint} if
$\mathcal{D}=\mathcal{D}^*$ and $L^*=L$.   Ignoring domain subtleties (as is
the case for bounded operators), 
self-adjoint operators are the analogue
of Hermitian matrices. 

\begin{definition} Given an operator $L$ on a Hilbert space,
an operator 
$M$ is said to be the inverse of $L$ if
the domain of $M$ is the range of $L$, the range of $M$ is
the domain of $L$, and
$ML=I$ and $LM=I$. 
\end{definition}

We will restrict ourselves now to operators whose $\mathcal{D}$ is the entire space $\mathcal{H}$.

An operator $L:\mathcal{H}\rightarrow\mathcal{H}$
is invertible if it 1) is one-to-one 2) is onto 3) has bounded
inverse. In finite-dimensional spaces, only the first condition is
required. 

An operator is bounded if it maps bounded inputs to bounded outputs. A
bounded operator is invertible if it is one-to-one and onto. 

\begin{definition}
The \emph{spectrum} of an operator $L$ on $\mathcal{H}$ is defined as
\[
\sigma(L)=\left\{\lambda \in \Complex \:\: \bigl|\:\: L-\lambda I \textnormal{ is not invertible
  }\right\}.
\]
\end{definition}

The spectrum 
of a bounded operator can be
partitioned into two classes, depending on the
reason that $L-\lambda I$ fails to be invertible.

A complex number $\lambda$ is considered part of 
the \emph{discrete spectrum} if $L-\lambda I$ is not one-to-one, \ie,
$L\psi=\lambda\psi$ has a non-zero solution $\psi\in H$.
In this case, $\lambda$ is called an \emph{eigenvalue},
and each $\psi$ satisfying this equation for the given $\lambda$
is an associated \emph{eigenvector}.
The set of all eigenvalues is called the \emph{point spectrum}
or \emph{discrete spectrum} of $L$, $\sigma_{\textnormal{disc}}(L)$.

It can also happen that $L-\lambda I$ fails to be surjective,
\ie, the range of $L-\lambda I$ is a proper subset of $\mathcal{H}$.
In this paper, we call the set of $\lambda$ for which this happens 
the \emph{continuous spectrum} $\sigma_{\textnormal{cont}}(L)$
\cite{reed1980mmm}. Some authors subdivide what we refer
to as the continuous spectrum into further classes 
(\eg, residual spectrum, essential spectrum, etc.) 
\cite{reed1980mmm};  however, for the purposes of this paper,
classification into discrete and continuous spectra will suffice.

In finite-dimensional Hilbert spaces the spectrum
is entirely discrete.
This may no longer be true in infinite-dimensional
spaces, where the eigenvalues (if they exist)
may only be one part of the spectrum.
The spectrum of a self-adjoint operator is real.

The following examples illustrate some of these possibilities.

\begin{example}
The operator $L(x(t))=tx(t)$, $x(t)\in L^2[0,1]$,
has no eigenvalues and its spectrum is purely 
continuous $\sigma(L)=[0,1]$.
\qed
\end{example}

\begin{example}
The Fourier transform operator
$\const{F}(q)(\omega)=\int\limits_{-\infty}^{\infty}q(t)e^{j\omega t}\der t$,
regarded as on operator on
$L^2(\Reals)$ has
the property $\const{F}^4=I$.
The eigenvalues are therefore the discrete values $\{\pm 1,\pm
j\}$. If $p(t)$ is an arbitrary polynomial,
then $p(t)\exp(-t^2/2)$ is an eigenfunction.
\qed
\end{example}

\begin{example}
Let $\nabla^2$ be the Laplace operator and let $r$ denote the radial
distance in the three-dimensional space. The operator $L=-\nabla^2-\frac{1}{r}$ is self-adjoint,
and therefore it has a real spectrum.
The continuous spectrum is [0,$\infty$) and the discrete 
spectrum is given by $\lambda_n=-\alpha\frac{1}{n^2}$, $n=1,2,\ldots$,
for some constant $\alpha$.
\qed
\end{example}

\begin{example}
It is possible that the discrete spectrum of an operator
is uncountable.
For example, for a sequence $(x_0,x_1,\ldots) \in
\ell^2 = \{ x: \sum_i|x_i|^2 < \infty \}$, let
the left-shift operator $L$ be
defined as 
$L\left(x_0,x_1,x_2,\ldots\right)=\left(x_1,x_2,\ldots \right)$.
The spectrum
consists of the unit disk $|\lambda|\leq 1$. The portion $|\lambda|<1$ is
the discrete spectrum while
$|\lambda|=1$ is the continuous spectrum.
The adjoint of $L$ is the right-shift operator
$R\left(x_0,x_1,x_2,\ldots\right)=\left(0,x_0,x_1,\ldots \right)$. This
operator has the same total spectrum, but it is entirely continuous.
\qed
\end{example}

Hermitian matrices are always diagonalizable, have real eigenvalues,
and possess a complete set of orthonormal eigenvectors,
which provide a basis for the input space.
There is a perfect generalization of this result to
\emph{compact} self-adjoint operators\cite{reed1980mmm}.   
\begin{theorem}[Hilbert-Schmidt spectral theorem] 
Let $L$ be a compact self-adjoint operator in $\mathcal{H}$. Then it
is always 
possible to find eigenvectors $\{\psi_n\}$
  of $L$ forming an orthonormal basis for
  $\mathcal{H}$. 
\end{theorem}
The spectral theory of operators that are not compact and
self-adjoint is more involved; see \cite{reed1980mmm}.

An important class of operators are the multiplication operators,
which are analogous to diagonal matrices.

\begin{definition}[Multiplication operator]
Let $f(t)$ be an arbitrary function.
The operator $L$ defined by $(L_f\psi)(t)=f(t)\psi(t)$,
which performs sample-wise multiplication, is called the
\emph{multiplication operator} or 
\emph{diagonal operator} induced by $f(t)$. 
\end{definition}

\begin{theorem}
Every bounded self-adjoint operator in a separable Hilbert space $\mathcal{H}$ is unitarily equivalent to a multiplication
operator $\Gamma$, \ie, $L=U\Gamma U^{-1}$ where $U$ is unitary.
\end{theorem}
It follows that the essence of a bounded self-adjoint operator is just a
multiplication operator. 

\section{Proof of Elementary Properties of the NFT}
\label{sec:proof-prop-nft}

In this section, we sketch the proofs of the properties
of the NFT stated in Section~\ref{subsec:elemprop}.
\begin{enumerate}
\item[1)] If $||q(t)||_{L^1}\ll 1$, then $y^2(t,\lambda)$ and $q^2(t)$
  terms can be
  ignored in \eqref{eq:qhat-ode} and \eqref{eq:a-ode}. From the
  resulting equations, it follows that there is no discrete spectrum
and $\hat{q}(\lambda)\rightarrow Q(\lambda)$.  The quadratic
terms are introduced by the NFT to account for the nonlinearity.
\item[2)] This follows from 1) in above, and that when $|a|\ll 1$, the
  squared terms representing the nonlinearity can be ignored.
\item[3)] This follows from replacing $q(t)$ with $e^{j\phi}q(t)$ in
  \eqref{eq:qhat-ode} and \eqref{eq:a-ode}. Alternatively, it can be proved 
  by changing variables $u=\diag(1,\exp(j\phi))v$ in $v_t=P(\lambda,qe^{j\phi})v$ to get
  $u_t=P(\lambda,q)u$ with the same boundary condition
  \eqref{eq:bound-cond-2}. Since $u_1=v_1$ and $u_2=\exp(j\phi)v_2$, $a(\lambda)$ does
  not change and $b(\lambda)$ is scaled by $e^{- j\phi}$.
\item[4)] This follows from noting that a) the eigenproblem \eqref{eq:dv-dt} is invariant under transformation
$t^\prime=t/a$, $q^\prime=aq$ and $\lambda^\prime=a\lambda$, and b) all boundary conditions are invariant
under this transformation since $\exp(\pm j\lambda^\prime t^\prime)=\exp(\pm j \lambda t)$. Note that if
  $\sgn(a)<0$, boundary conditions at $\pm\infty$ are interchanged.

\item[5, 6)] Property 5) and 6) follow by replacing $q(t)$ with $q(t-t_0)$ and 
$e^{-2j\omega t}q(t)$ in \eqref{eq:qhat-ode} and \eqref{eq:a-ode}, and
accordingly changing variables. Alternatively, Property 5) follows by noting the structure of the
propagator in Property \ref{item:propagator}) \emph{below}. Property 6) can also be  proved by noting that if 
 $v_t=P(\lambda, qe^{-2j\omega t})v$ and  $v(t\rightarrow\infty,\lambda)\rightarrow (1,0)^T\exp(-j\lambda t)$, then
$u=\diag(\exp(j\omega t), \exp(-j\omega t))v$
satisfies $u_t=P(\lambda^\prime, q)u$ and $u(t\rightarrow\infty,\lambda)\rightarrow (1,0)^T\exp(-j\lambda^\prime t)$, for 
$\lambda^\prime=\lambda-\omega$.

\item[7)] This is the statement of \eqref{eq:nft-evol}.
\item[8)] The following identity, known as the trace formula, can be easily proved for the
 nonlinear Fourier transform \cite{ablowitz2006sai}
\begin{IEEEeqnarray*}{l}
c_n=\frac{4}{n}\sum\limits_{i=1}^{N}\Im(\lambda_i^n)
+
\frac{1}{\pi}\int\limits_{-\infty}^{\infty}\lambda^{n-1}\log\left(1+|\hat{q}(\lambda)|^2\right)\der \lambda.
\end{IEEEeqnarray*}
Here $c_n$ are the secondary constants of motion, \ie, quantities,
directly in terms of the time domain data, which are preserved during the flow of the NLS equation. The first few
ones are the energy
\begin{IEEEeqnarray*}{rCl}
c_1&=&\int\limits_{-\infty}^{\infty}|q(t)|^2\der t,
\end{IEEEeqnarray*}
momentum
\begin{IEEEeqnarray*}{rCl}
c_2&=&\frac{1}{2j}\int\limits_{-\infty}^{\infty} q(t)q^*_t(t)\der t,
\end{IEEEeqnarray*}
and the Hamiltonian
\begin{IEEEeqnarray*}{rCl}
c_3&=&-\frac{1}{4}\int\limits_{-\infty}^{\infty}
\left(|q(t)|^4-|q_t(t)|^2\right)\der t.
\end{IEEEeqnarray*}
Parseval's identity is the trace formula at $n=1$. 
\end{enumerate}

The other two properties related to NFT can be proved too.
\begin{enumerate}
\item 
\label{item:propagator}
Let $q(t)$ be supported in the interval $[t_1,t_2]$ and let
  $K(q(t),t_1,t_2)$ denote a propagator (linear transformation) which
  maps $v(t,\lambda)$ in \eqref{eq:dv-dt} from $t=t_1$ to
  $t=t_2$, \ie, $v(t_2,\lambda)=Kv(t_1,\lambda)$. The propagator is structured as
\begin{IEEEeqnarray*}{rCl}
K=  
\begin{pmatrix}
    a(\lambda)e^{-j\lambda (t_2-t_1)} & -b^*(\lambda^*)e^{-j\lambda (t_2+t_1)}\\
b(\lambda)e^{j\lambda (t_2+t_1)} & a^*(\lambda^*)e^{j\lambda (t_2-t_1)}
  \end{pmatrix}.
\end{IEEEeqnarray*}
Let $q_1(t)$ and $q_2(t)$, supported, respectively, in the intervals $[t_1,t_2]$ and
$[t_2,t_3]$, $t_1<t_2<t_3$, correspond to the propagators
$K_1(q_1(t),t_1,t_2)$ and $K_2(q_2(t),t_2,t_3)$. Then $q(t)=q_1(t)+q_2(t)$ is supported in
$[t_1,t_3]$, and from linearity corresponds to the propagator $K=K_2K_1$ 
\begin{IEEEeqnarray*}{rCl}
K&=&  
\begin{pmatrix}
    a_2(\lambda)e^{-j\lambda (t_3-t_2)} & -b_2^*(\lambda^*)e^{-j\lambda (t_3+t_2)}\\
b_2(\lambda)e^{j\lambda (t_3+t_2)} & a_2^*(\lambda^*)e^{j\lambda (t_3-t_2)}
  \end{pmatrix}
\\
&&\times \:
\begin{pmatrix}
    a_1(\lambda)e^{-j\lambda (t_2-t_1)} & -b_1^*(\lambda^*)e^{-j\lambda (t_2+t_1)}\\
b_1(\lambda)e^{j\lambda (t_2+t_1) } & a_1^*(\lambda^*)e^{j\lambda (t_2-t_1)}
  \end{pmatrix}.
  \end{IEEEeqnarray*}
The $1\times 1$ and $2\times 1$ elements are   
\begin{IEEEeqnarray*}{rCl}
 K_{11} &=& e^{-j\lambda (t_3-t_1)}\left( a_1(\lambda)a_2(\lambda)-b_1(\lambda)b_2^*(\lambda^*)\right),\\
 K_{21} &=& e^{j\lambda (t_3+t_1)}\left( a_1(\lambda)b_2(\lambda)+b_1(\lambda)a_2^*(\lambda^*)\right).
\end{IEEEeqnarray*}
Comparing with $K_{11}= a(\lambda)e^{-j\lambda (t_3-t_1)}$ and
$K_{21}=b(\lambda)e^{j\lambda (t_3+t_1)}$, we get the desired
result.  
\item See \cite{klaus2003ezss}.
\end{enumerate}

\section{Riemann-Hilbert Factorization Problem}

\label{app:rhp}
Recall that a complex function $f(z)=u(x,y)+jv(x,y)$ in the complex plane
$z=(x,y)$ is differentiable at a point $(x_0,y_0)$ if and only if
the partial derivatives $u_x$, $u_y$, $v_x$ and $v_y$ are continuous and
\begin{IEEEeqnarray}{rCl}
u_x=v_y, \quad u_y=-v_x.
\label{eq:cauchy-riemann}
\end{IEEEeqnarray}  
The compatibility conditions \eqref{eq:cauchy-riemann} are called
the Cauchy-Riemann conditions and are obtained by equating the
limit $\Delta z \rightarrow 0$ in the definition of derivative along the real and imaginary axes.
This means that, unlike real
functions, differentiability of a complex function imposes a constraint
between the real and imaginary parts of the function. 

A function $f(z)$ is said to be analytic at $(x_0,y_0)$ if it is
differentiable in a neighborhood of that point. If $f(z)$ is analytic in
an open region $\Omega$ of the complex plane, then a
power series 
representation of $f(z)$ is convergent in $\Omega$. Existence of 
a power series
representation has a number of interesting
consequences. For instance, the set $S$ of the zeros of a function
analytic in a nonempty connected 
open subset $\Omega$ of $\Complex$ is either the entire $\Omega$ or 
has no limit point in $\Omega$ (\ie, $S$ consists of countable isolated points in
$\Omega$) \cite{stein2003ca}.

\subsection{The Scalar Riemann-Hilbert Problem}
In the scalar Riemann-Hilbert (RH) problem, the task is to find
functions $f^+(z)$ and $f^-(z)$, analytic, respectively, inside and
outside of a given smooth closed contour 
$C$, such that on $C$
\begin{IEEEeqnarray}{rCl}
f^+(t)=g(t)f^-(t)+h(t),\quad t\in C,
\label{eq:srh}
\end{IEEEeqnarray}
where $h(t)$ and $g(t)$ (with
$g(t)\neq 0$ for all $t\in C$), are given functions satisfying a H\"older
condition on $C$. 

A solution of this problem can be found in \cite{ablowitz2003cvi,
  musk2008}. For a brief review, in this section  we follow
\cite{ablowitz2003cvi}. 

The following lemma is central in solving the simplest scalar RH factorization.  

\begin{lemma}[Sokhotski-Plemelj formulae]
Let $C$ be any smooth, closed, counter-clockwise, contour in the complex plane
and let $f(x)$ be any function satisfying a H\"older condition on $C$
defined by
\begin{IEEEeqnarray*}{rCl}
|f(t)-f(\tau)|\leq k|t-\tau|^\alpha, \, k>0, \, \forall t,\tau\in C ,
\IEEEeqnarraynumspace
\end{IEEEeqnarray*}
for some $0<\alpha\leq 1$. Then the projection integral
\begin{IEEEeqnarray}{rCl}
F(\zeta)=\frac{1}{2\pi j}\oint\limits_{C} \frac{f(z)}{z-\zeta}\der z \label{eq:projection}
\end{IEEEeqnarray}
is analytic everywhere in $\Complex$ except possibly at points $\zeta$
on the contour $C$ (where $F(\zeta)$ is not defined).
If $\zeta$ approaches $C$ along a path entirely inside the contour $C$, then
\begin{IEEEeqnarray}{rCl}
F^+(\zeta)\eqdef \lim\limits_{z\rightarrow\zeta}F(\zeta)=\frac{f(\zeta)}{2}+\frac{1}{2\pi
  j}\textnormal{ p.v. }\int\limits_{C} \frac{f(z)}{z-\zeta}\der z.
\IEEEeqnarraynumspace
\label{eq:+lim}
\end{IEEEeqnarray}
If $\zeta$ approaches $C$ along a path entirely outside the contour $C$, then
\begin{IEEEeqnarray}{rCl}
F^{-}(\zeta)\eqdef \lim\limits_{z\rightarrow\zeta}F(\zeta)=-\frac{f(\zeta)}{2}+\frac{1}{2\pi j}\textnormal{ p.v. }\int\limits_{C} \frac{f(z)}{z-\zeta}\der z.
\IEEEeqnarraynumspace
\label{eq:-lim}
\end{IEEEeqnarray}
Here $\textnormal{p.v.}$ denotes the principal value integral defined by 
\begin{IEEEeqnarray*}{rCl}
\textnormal{ p.v. }\int\limits_{C} \frac{f(z)}{z-\zeta}\der z=
\lim\limits_{\epsilon \rightarrow 0} \oint_{C-C_\epsilon} \frac{f(z)}{z-\zeta}\der z,
\end{IEEEeqnarray*}
in which $C_\epsilon$ is an infinitesimal part of $C$ centered at
$z=\zeta$ and with length $2\epsilon$.
\label{lem:plemelj}
\end{lemma}

\begin{IEEEproof}
See \cite{ablowitz2003cvi, musk2008}.
\end{IEEEproof}

The projected function $F(\zeta)$ is a \emph{sectionally analytic} function
of $\zeta$ with respect to $C$, \ie, it is analytic in sections
$C^+$ (the interior of $C$) and $C^-$ (the exterior of $C$), and the limits
$F^\pm(\zeta)$ exist (as given by \eqref{eq:+lim} and \eqref{eq:-lim}).

A consequence of Lemma~\ref{lem:plemelj} is that $F^\pm (\zeta)$ satisfy
the following \emph{jump condition} on the boundary $C$
\begin{IEEEeqnarray*}{rCl}
F^+(t)-F^-(t)=f(t).
\end{IEEEeqnarray*}
The projection operator therefore produces functions which are analytic almost
everywhere, except on a contour where it experiences a jump in its limits.

Both unknowns $f^+(z)$ and $f^-(z)$ in the scalar RH problem can be determined from the single boundary
equation \eqref{eq:srh}, using the projection operator
\eqref{eq:projection}. To see this, first consider the homogeneous
case where $h(t)=0$. One can rewrite \eqref{eq:srh} as a jump condition 
\begin{IEEEeqnarray*}{rCl}
\log f^+(z)-\log f^-(z)=\log g(z),\quad z\in C.
\end{IEEEeqnarray*}
Functions $\log f^+(z)$ and $\log f^-(z)$ can be viewed as 
portions of a single sectionally analytic function $\log f(z)$ which is analytic in
$C^+$ and $C^-$ and on boundary $C$ its limits jump as $\log g(t)$.
In view of the projection operator $P$ \eqref{eq:projection}, consider
\begin{IEEEeqnarray*}{rCl}
 \log f(z)=\frac{1}{2\pi j}\oint\limits_{C}\frac{\log g(\lambda)}{\lambda -z}\der \lambda.
\end{IEEEeqnarray*}
If $\log g(t)$ satisfies a H\"older condition on $C$, then
$\log f(z)$
is analytic strictly inside and outside $C$.
On $C$, we can define $\log f^+(z)$ and
$\log f^-(z)$, respectively, as equal
to the limits \eqref{eq:+lim} and \eqref{eq:-lim}. The
function obtained in this way satisfies \eqref{eq:srh} and has the desired
analyticity properties. 

Note however that, unlike $g(t)$, $\log g(t)$ in the integrand may not
satisfy a H\"older condition. To resolve this issue, 
we can multiply $g(t)$ by a decaying factor $t^{-k}$, for a suitable $k$, 
to make $t^{-k}g(t)$ H\"older, and obtain $f^+(t)=\left(t^{-k}
  g(t)\right)(t^kf^-(z))$. Therefore, defining
\begin{IEEEeqnarray*}{rCl}
F(z)&=&\exp\left(\frac{1}{2\pi j}\oint\limits_{C}\frac{\log
    \lambda^{-k} g(\lambda)}{\lambda -z}\der \lambda\right),
\end{IEEEeqnarray*}
we have the following solution for the homogeneous RH problem:
\begin{IEEEeqnarray*}{rCl}
f^+(z)=
\begin{cases}
F(z), & z\in C^+,\\
F^+(z), & z\in C,
\end{cases}
\end{IEEEeqnarray*}
and
\begin{IEEEeqnarray*}{rCl}
f^-(z)=
\begin{cases}
z^{-k}F(z), & z\in C^-,\\
z^{-k}F^{-}(z), & z\in C.
\end{cases}
\end{IEEEeqnarray*}
Here $k$ can be chosen so that $t^{-k}g(t)$ is continuous
and the total phase change of $\log t^{-k} g(t)$ is zero along
the closed path $C$. 

The solution $f^{\pm}(z)$ is called the fundamental solution
to the scalar RH problem.
From the
homogeneity of \eqref{eq:srh}, one can obtain other solutions by 
multiplying $f^{\pm}(z)$ by any entire function in $C$. 

We can generalize the above procedure to solve the non-homogeneous
Riemann-Hilbert problem \eqref{eq:srh}. In this case, we can find a
factorization $g(t)=g^+(t)/g^-(t)$ by solving a homogeneous
Riemann-Hilbert problem with 
boundary conditions $g^+(t)=g(t)g^-(t)$. Then \eqref{eq:srh} is reduced to 
\begin{IEEEeqnarray*}{rCl}
\frac{f^+(t)}{g^+(t)}-\frac{f^-(t)}{g^-(t)}=\frac{h(t)}{g^+(t)},
\end{IEEEeqnarray*}
which, as before, can be solved in closed form using the Plemelj formulae.

\subsection{The Matrix Riemann-Hilbert Problem}
When formulating the inverse nonlinear Fourier transform, we face a matrix Riemann-Hilbert
problem \eqref{eq:mrh}. Matrix RH problems are generally more involved
and may not allow closed-form solutions \cite{musk2008}. As we will see in the Appendix~\ref{sec:sol-of-rh}, for
the particular matrix RH problem \eqref{eq:mrh}, the projection operator \eqref{eq:projection} is
sufficient to solve the problem.  

\section{Proof of Lemma~\ref{lem:analyticity}}
\label{sec:analyticity}
Analyticity of the canonical eigenvectors is directly a property of the
Zakharov-Shabat system \eqref{eq:dv-dt}. A formal proof can be found
in \cite{ablowitz2006sai, ablowitz2003dcn}; we simplify and outline it in this
section for the sake of completeness. 

Consider the scaled canonical eigenvectors 
\eqref{eq:scaled-eigvecs}. Transforming the Zakharov-Shabat system \eqref{eq:dv-dt}, the scaled eigenvectors satisfy
\begin{IEEEeqnarray}{lCrlCr}
 V^2_t&=&\begin{pmatrix}0 & q\\ -q^* & 2j\lambda \end{pmatrix}V^2, \quad V^2(-\infty)&=&\begin{pmatrix}1\\0\end{pmatrix},  
\label{eq:V2}
\\
\tilde{V^2_t}&=&\begin{pmatrix}-2j\lambda & q\\ -q^* & 0 \end{pmatrix}\tilde{V^2},  
\quad \tilde{V^2}(-\infty)&=& \begin{pmatrix}0\\1\end{pmatrix}, \nonumber
\\
 \tilde{V^1_t}&=&\begin{pmatrix}0 & q\\ -q^* & 2j\lambda \end{pmatrix}\tilde{V^1}, \quad \tilde{V^1}(\infty)&=&\begin{pmatrix}1\\0\end{pmatrix},  \nonumber
\\
 V^1_t&=&\begin{pmatrix}-2j\lambda & q\\ -q^* & 0 \end{pmatrix}V^1,  \quad V^1(\infty)&=&\begin{pmatrix}0\\1\end{pmatrix}. \nonumber
\end{IEEEeqnarray} 

Let us, for instance, solve for the canonical eigenvector $V^2$ in
\eqref{eq:V2}. Considering the $q$ terms as an external force and
using the Duhamel's formula \cite{tao2006nde}, \eqref{eq:V2} can be
transformed into its integral representation
\begin{multline}
\label{eq:V2-integral}
V^2(t,\lambda)=\begin{pmatrix}1\\0\end{pmatrix}\\
+\int\limits_{-\infty}^{\infty}h(t-t^\prime,\lambda)\begin{pmatrix}0
  & q(t^\prime)\\ -q^*(t^\prime) &
  0\end{pmatrix}V^2(t^\prime,\lambda)\der t^\prime, 
\end{multline}
where the system ``impulse response'' $h(t,\lambda)$ is 
\begin{IEEEeqnarray}{rCl} 
h(t,\lambda)=\begin{pmatrix} u(t) & 0 \\0 & e^{2j\lambda t}u(t) \end{pmatrix},
\label{eq:impulse-resp} 
\end{IEEEeqnarray} 
where $u(t)$ is the step function and we have ignored transient terms since the boundary
condition starts at $t=-\infty$.

The analyticity of eigenvectors can be seen intuitively at
at this stage. The impulse response \eqref{eq:impulse-resp}
involves the term $e^{2j\lambda t}u(t)$
and hence \eqref{eq:V2} is well defined in $\Complex^+$ if
$q\in L^1(\Reals)$. The impulse 
response for the $V^1$ equation, by converting it to the $V^2$ equation, involves $-e^{-2j\lambda t}u(-t)$ and
hence it is bounded in the same region. The impulse response for
$\tilde{V}^2$ and $\tilde{V}^1$ have terms proportional to
$e^{-2j\lambda t}u(t)$ and $-e^{2j\lambda t}u(-t)$, respectively, and
therefore these eigenvectors are analytic in $\Complex^-$ for $q\in L^1(\Reals)$. 

A more precise argument proceeds by solving \eqref{eq:V2-integral}
explicitly. Duhamel's integral \eqref{eq:V2-integral} is of the form of a
fixed-point map
\[
V^2=\begin{pmatrix} 1\\0 \end{pmatrix}+T(V^2),
\]
where $T$ is the linear operator underlying the integral term in \eqref{eq:V2-integral}. A
candidate for the solution is the sum 
\begin{equation}
\label{eq:V2-sol}
V^2=\sum\limits_{k=0}^{\infty} U^k,
\end{equation}
where $U^k$ satisfy the iteration
\[
U^{k+1}=\begin{pmatrix} 1\\0\end{pmatrix}+T(U^k), \quad U_0=\begin{pmatrix} 0\\0\end{pmatrix}.
\]
The first few terms are
\begin{IEEEeqnarray}{rCl}
U^1&=&\begin{pmatrix}1\\0\end{pmatrix},\quad
U^2=\begin{pmatrix}0\\-\int\limits_{t_1=-\infty}^{t}q^*(t_1)e^{2j\lambda(t-t_1)}\der t_1\end{pmatrix},\label{eq:U1U2}\IEEEeqnarraynumspace \\
U^3&=&\begin{pmatrix}-\int\limits_{t_2=-\infty}^{t}\int\limits_{t_1=-\infty}^{t_2}q(t_2)q^*(t_1)e^{2j\lambda(t_2-t_1)}\der t_1
  \der t_2\\ 0\end{pmatrix},
\IEEEeqnarraynumspace
\label{eq:U3} 
\end{IEEEeqnarray}
and the $k^{\rm{th}}$ term is recursively defined by
\begin{IEEEeqnarray}{rCl}
U^{k+1}=\begin{pmatrix}1 \\0 \end{pmatrix}+\int\limits_{-\infty}^{t}\begin{pmatrix}q(t^\prime)U^k_2(t^\prime,\lambda)\\
-q^*(t^\prime)U^k_1(t^\prime,z)e^{2j\lambda (t-t^\prime)}\end{pmatrix}\der t^\prime.\IEEEeqnarraynumspace
\label{eq:Uk} 
\end{IEEEeqnarray}  

By induction on $k$, as \eqref{eq:Uk} suggests, if $U^k$ is analytic and
$q(t)\in L^1(\mathbb{R})$, then
$U^{k+1}$ is analytic. Since, the series \eqref{eq:V2-sol} is
uniformly convergent on $t$, $V^2$ is analytic in $\Complex^+$. 

Similarly one proves the analyticity of the other canonical eigenvectors in
their corresponding region.

\section{Asymptotic Form of Canonical Eigenvectors and Nonlinear
Fourier Coefficients when $|\lambda|\gg 1$}
\label{app:asymp-eigenvectors}

If $\lambda\in C^+$ and $|\lambda|\gg 1$, then
$\frac{1}{j(\omega-2\lambda)}=\frac{-1}{2j\lambda}+O(\lambda^{-2})$ and, taking the
inverse Fourier transform, we can approximate $e^{-2j\lambda(t_1-
  t)}u(t-t_1)=-\frac{1}{2j\lambda}\delta(t_1-t)+O(\lambda^{-2})$. Substituting
into \eqref{eq:V2-sol}, \eqref{eq:U1U2}, \eqref{eq:U3}, for $|\lambda|\gg 1$ we obtain 
\begin{IEEEeqnarray}{rCl} 
V^2(t,\lambda)=
\begin{pmatrix}
  1+\frac{1}{2j\lambda}\int\limits_{-\infty}^{t}|q(t)|^2 \der t \\
  \frac{1}{2j\lambda}q^*(t) \end{pmatrix}+O(\lambda^{-2}).
\label{eq:V2-asym}
\end{IEEEeqnarray}  
A similar asymptotic expression can be derived for $V^1$ ($\lambda\gg 1$)
\begin{IEEEeqnarray*}{rCl} 
V^1(t,\lambda)=\begin{pmatrix}
\frac{1}{2j\lambda}q(t)\\
  1+\frac{1}{2j\lambda}\int\limits_{t}^{\infty}|q(t)|^2 \der t 
\end{pmatrix}+O(\lambda^{-2}).
\end{IEEEeqnarray*}  

For the nonlinear Fourier coefficients, if $|\lambda|\rightarrow \infty$, $q$ can be assumed zero in \eqref{eq:dv-dt}
compared to $j\lambda$. Thus $v(t,\lambda)$ approaches the boundary
conditions at $t=\pm\infty$. Therefore
\begin{IEEEeqnarray}{rCl}
a(\lambda)&=&\inners{v^2(t,\lambda)}{v^1(t,\lambda)}\nonumber\\
&=&
\inners{v^2(+\infty,\lambda)}{v^1(+\infty,\lambda)}\nonumber\\
&\rightarrow&
\inners{v^2(-\infty,\lambda)}{v^1(+\infty,\lambda)}=1.
\label{eq:a-asym}
\end{IEEEeqnarray} 
Similarly, it is shown that $b(\lambda)\rightarrow 0$ as
$|\lambda|\rightarrow\infty
$.
\section{Solution of the Riemann-Hilbert problem}
\label{sec:sol-of-rh}
In Section~\ref{sec:inft}, the inverse nonlinear Fourier transform was
formulated as an instance of the Riemann-Hilbert factorization
problem. Following the discussion in the Appendix~\ref{app:rhp}, the resulting factorization problem can be solved in a
simplified manner via an appropriate contour integration. 

Dividing both sides of the projection equations \eqref{eq:proj-1}--\eqref{eq:proj-2} by
$a(\lambda)(\lambda-\zeta)$, for parameter $\zeta\in C^-$, and integrating on the real axis $-\infty<\lambda<\infty$, we obtain
\begin{multline}
\frac{1}{2\pi j}\int\limits_{\lambda=-\infty}^{\infty}\frac{V^2(t,\lambda)}{a(\lambda)(\lambda-\zeta)}\der \lambda=\frac{1}{2\pi
  j}\int\limits_{\lambda=-\infty}^{\infty}\frac{\tilde{V^1}(t,\lambda^*)}{\lambda-\zeta}\der \lambda\\
+
\frac{1}{2\pi j}\int\limits
_{\lambda=-\infty}^{\infty}\frac{\hat{q}(\lambda)e^{2j\lambda
    t}V^1(t,\lambda)}{\lambda-\zeta}\der \lambda, \IEEEeqnarraynumspace 
\label{eq:integ-proj-1}
\end{multline}
in which the integration is performed on the open path
$z=\lambda$, $-\infty<\lambda<\infty$. The integration
path thus passes the singularity $\lambda=\zeta$ from
above in all the integrals.

Cauchy integrals in \eqref{eq:integ-proj-1} are computed from the
residue theorem. The integration path $-\infty<\lambda<\infty$ can be 
closed in the upper or lower half-planes. To compute the first integral, we
close the path in the
upper half-plane and denote the resulting closed contour by $C_{-\zeta}^{+}$, \ie, the upper half
plane and excluding the singularity $z=\zeta$
\begin{IEEEeqnarray}{lCl}
&&\frac{1}{2\pi  j}\int\limits_{-\infty}^{\infty}\frac{V^2(t,\lambda)}{a(\lambda)(\lambda-\zeta)}\der\lambda \nonumber\\
&=&\frac{1}{2\pi
  j}\oint_{C_{-\zeta}^{+}}\frac{V^2(t,z)}{a(z)(z-\zeta)}\der z -\lim\limits_{R\rightarrow \infty}\frac{1}{2\pi j}\int\limits_{Re^{j0}}^{Re^{j\pi}}\frac{\begin{pmatrix}1\\0\end{pmatrix}}{a(z)(z-\zeta)}\der z
\nonumber\\
&=&\sum\limits_{i=1}^{N}\frac{V^2(t,\lambda_j)}{a_{\lambda}(\lambda_j)(\lambda_j-\zeta)}
-\frac{1}{2}\begin{pmatrix}1\\0\end{pmatrix}
\nonumber\\
&=&\sum\limits_{i=1}^{N}\frac{b(\lambda_i)e^{2j\lambda_i
t}V^1(t,\lambda_i)}{a_{\lambda}(\lambda_i)(\lambda_i-\zeta)}-\frac{1}{2}\begin{pmatrix}1\\0\end{pmatrix}
\nonumber\\
&=&\sum\limits_{i=1}^{N}\frac{\tilde{q}(\lambda_i)e^{2j\lambda_i
    t}V^1(t,\lambda_i)}{\lambda_i-\zeta}-\frac{1}{2}\begin{pmatrix}1\\0\end{pmatrix},
\label{eq:first-integal}
\end{IEEEeqnarray}
where, in the second line, when $R\rightarrow\infty$, we have used the asymptotic values
\eqref{eq:V2-asym} and \eqref{eq:a-asym} in Appendix~\ref{app:asymp-eigenvectors}. Note that we assumed that eigenvalues $\lambda_j$ are all \emph{simple} zeros of $a(\lambda)$, \ie, no multiplicity.

To compute the second integral in \eqref{eq:integ-proj-1}, we close the integration path
in the lower half-plane and denote the resulting closed contour by
$C_{+\zeta}^{-}$, \ie, the lower half-plane and including the singularity $z=\zeta$ 
 \begin{multline}
\frac{1}{2\pi j}\int\limits_{\lambda=-\infty}^{\infty}\frac{\tilde{V^1}(t,\lambda^*)}{\lambda-\zeta}\der \lambda
=
\frac{1}{2\pi j}\oint_{C_{+\zeta
}^-}\frac{\tilde{V^1}(t,z^*)}{z-\zeta}\der z\\
-
\lim\limits_{R\rightarrow \infty}\frac{1}{2\pi
  j}\int\limits_{Re^{j2\pi}}^{Re^{j\pi}}\frac{\begin{pmatrix}1\\0\end{pmatrix}}{z-\zeta}\der z
=-\tilde{V^1}(t,\zeta^*)+\frac{1}{2}\begin{pmatrix}1\\0\end{pmatrix}.
\label{eq:second-integral}
\end{multline}

The last integral in \eqref{eq:integ-proj-1} is not computed, because the boundedness of $e^{2j\lambda
t}$ depends on the sign of $t$.  For $t>0$, we can consider $C_{-\zeta}^+$ which leads
to the expression \eqref{eq:first-integal} multiplied by $u(t)$. For $t<0$, we should
inevitably consider $C_{+\zeta}^-$, where poles of $V^1(x,\lambda)$ are
unknown. As a result, this integral is left untreated. 

Using \eqref{eq:first-integal} and \eqref{eq:second-integral} in
\eqref{eq:integ-proj-1}, we obtain an integral equation relating
canonical eigenvectors $V^1$ and $\tilde{V^1}$ to $\hat{q}(\lambda)$ and 
$\tilde{q}(\lambda_j)$
\begin{IEEEeqnarray}{rCl}
\tilde{V^1}(t,\zeta^*)&=&\begin{pmatrix}1\\0\end{pmatrix}+\sum\limits_{i=1}^{N}\frac{\tilde{q}(\lambda_i)e^{2j\lambda_i
    t}V^1(t,\lambda_i)}{\zeta-\lambda_i}\nonumber\\
&&+\:
\frac{1}{2\pi j}\int\limits_{\lambda=-\infty}^{\infty}\frac{\hat{q}(\lambda)e^{2j\lambda
    t}V^1(t,\lambda)}{\lambda-(\zeta-j\epsilon)}\der \lambda.
\label{eq:rh-integ-eq-0}
\end{IEEEeqnarray}
This is equation \eqref{eq:rh-integ-eqs-a} in the Riemann-Hilbert system
 when $\zeta$ approaches the real line from below. Since \eqref{eq:rh-integ-eq-0} holds for
any $\zeta\in C^-$, evaluating \eqref{eq:rh-integ-eq-0} at
$\zeta=\lambda_j^*$, $j=1,\ldots,N$, produces
\eqref{eq:rh-integ-eqs-c}. The remaining equations are obtained by
taking the tilde ``$\sim$'' operation from these two equations and subsequently
replacing $\lambda$ and $\lambda_m$ with, respectively, $\lambda^*$
and $\lambda^*_m$.

\begin{IEEEbiographynophoto}{Mansoor I. Yousefi} received the Ph.D. degree from
the University of Toronto and the M.Sc.\ degree (with
honors) from
Sharif University of Technology, Tehran, Iran, both in
in Electrical  Engineering.
His research interests include information and communication theory, coding
theory, optimization, optical communication systems and machine intelligence.
He has won numerous awards and scholarships, including the IEEE Jack Keil Wolf
ISIT Student Paper Award in 2013, the Edward S. Rogers Scholarship, the
McAllister
Graduate Fellowship, the Shahid U. H. Qureshi Prize and the J. L. Allen Yen
Scholarship at the University of Toronto, and the Electrical Engineering
Department Award at Sharif University of Technology.
\end{IEEEbiographynophoto}

\begin{IEEEbiographynophoto}{Frank R. Kschischang}
received the B.A.Sc. degree (with honors) from the
University of British Columbia, Vancouver, BC, Canada, in 1985 and the M.A.Sc.
and Ph.D. degrees from the University of Toronto, Toronto, ON, Canada, in 1988
and 1991, respectively, all in electrical engineering.  He is a Professor of
Electrical and Computer Engineering at the University of Toronto, where he has
been a faculty member since 1991. During 1997--98, he was a visiting scientist
at MIT, Cambridge, MA; in 2005 he was a visiting professor at the ETH, Zurich,
and in 2011 and again in 2012--13 he was a visiting Hans Fischer Senior Fellow
at the Institute for Advanced Study at the Technical University of Munich.  

His research interests are focused primarily on the area of channel coding
techniques, applied to wireline, wireless and optical communication systems and
networks.  In 1999 he was a recipient of the Ontario Premier's Excellence
Research Award and in 2001 (renewed in 2008) he was awarded the Tier~I Canada
Research Chair in Communication Algorithms at the University of Toronto.  In
2010 he was awarded the Killam Research Fellowship by the Canada Council for
the Arts.  Jointly with Ralf K{\"o}tter he received the 2010 Communications
Society and Information Theory Society Joint Paper Award.  He is a recipient of
the 2012 Canadian Award in Telecommunications Research.  He is a Fellow of
IEEE, of the Engineering Institute of Canada, and of the Royal Society of
Canada.
  
During 1997--2000, he served as an Associate Editor for Coding Theory for the
\textsc{IEEE Transactions on Information Theory}, and since January 2014,
he serves as
this journal's Editor-in-Chief.   He also served as technical program co-chair
for the 2004 IEEE International Symposium on Information Theory (ISIT),
Chicago, and as general co-chair for ISIT 2008, Toronto.  He served as the 2010
President of the IEEE Information Theory Society.
\end{IEEEbiographynophoto}

\end{document}